%% file: main.tex
\documentclass[12pt]{article}
\usepackage[utf8]{inputenc}
\usepackage[round]{natbib}
\usepackage{amssymb,amsmath,amsthm}
\usepackage{float}
\usepackage{setspace}
\usepackage[margin=1in]{geometry}
\usepackage[hidelinks]{hyperref}
\usepackage{graphicx}
\usepackage{booktabs}
\usepackage{tabularx}
\usepackage[title]{appendix}
\counterwithout{figure}{section}
\counterwithout{table}{section}
\counterwithout{table}{subsection}

\newtheorem{proposition}{Proposition}

\title{Cultural Transmission, Property Rights, and Treatment of the Elderly}
\author{Matthew J. Baker \\ Department of Economics \\
Hunter College and the Graduate Center, CUNY \\ \\ and \\ \\ Joyce P. Jacobsen \\ Department of Economics \\ Hobart and William Smith Colleges \\ and Wesleyan University}

\begin{document}

\maketitle 

\begin{abstract}
We examine how production and the development of property rights interact with cultural transmission to shape the treatment of the elderly across societies. Our model posits that respect for the elderly arises endogenously: parents invest in cultivating cultural values in their children, who later reciprocate in proportion to this investment. We show that this model is functionally equivalent to one in which cultural goods are transferred by the elderly. We focus on the distinct roles of property rights, finding that while insecure output rights may promote elderly welfare, secure rights over productive resources can have comparable benefits. The model reveals a nonlinear relationship between cultural sophistication, property rights, and economic factors such as the capital and land intensity of production, driving variations in elderly well-being across societies. Finally, we consider how the model suggests demographic, technological, and policy changes influence elderly well-being across the spectrum of development.\end{abstract}

\thispagestyle{empty}
\newpage
\clearpage
\pagenumbering{arabic}

\onehalfspacing

\section{Introduction}

Aging and the elderly have always been part of human existence. Living into old age is common among hunter-gatherers and was likely even common among early Homo sapiens, suggesting that long lifespans may be an evolutionary adaptation to the hunter-gatherer lifestyle \citep{gurven09, kaplanrobson02, kaplanrobson03}. However, the experience of aging varies significantly across societies. In modern economies, the elderly often sustain themselves through personal means and government-provided benefits \citep{un19}. By contrast, in traditional societies, including preliterate ones, the elderly's roles and treatment are diverse. They may be revered for their knowledge, act as cultural custodians, or benefit from specialized or evolving property rights. In some cases, their longevity enables them to accumulate productive assets and experience.

Demographic and technological shifts are reshaping institutions and cultural practices surrounding aging. Increased longevity and declining fertility have dramatically expanded the elderly population worldwide, both in total and relative to younger ages. The global population aged 65 or older is expected to rise from 524 million in 2010 to 1.5 billion by 2050 \citep{who11}. In 1950, there were three children under five for every person over 65; by 2050, this ratio is projected to reverse \citep{who11}. These changes influence work, consumption, and savings patterns, particularly for the elderly. Popular works, such as Diamond (2012, Chapter 6), highlight aspects of this transformation, and some governments are responding with targeted policies \citep{un19}.\footnote{See \url{https://theweek.com/articles/462230/how-elderly-are-treated-around-world}.}

Across all economic development levels, but especially in less developed economies, the elderly play a key role in preserving and transmitting culture. This raises critical questions: how might demographic and technological changes affecting the elderly influence cultural transmission? \citep[p. 5]{who11}. Notably, much of the anticipated growth in the elderly population will occur in lower income countries, where tradition and culture strongly shape their condition \citep{who11, un19}. 

We present a simple model to examine the factors shaping elderly treatment across societies, from hunter-gatherers to modern economies. The model shows how material conditions shape cultural practices supporting the elderly and how technological and demographic shifts transform these practices. It highlights the interaction between property rights over productive assets and output, elderly treatment, and cultural transmission. A key result is the emergence of what anthropologists refer to as a `curvilinear' pattern: societies midway between traditional and modern economies are predicted to have the strongest cultural traditions of reverence for the elderly.

The rest of this paper is organized as follows. In Section \ref{lit}, we relevant models from the literature on the economics of aging, and how they relate to property rights, cultural transmission, and modernization. In Section \ref{ethnosec}, we discuss aspects of the cross-cultural literature on aging and treatment of the elderly.  In Section \ref{theorysec}, we present our model, which emphasizes property rights, and shows how a `curvilinear' pattern in cultural concern for the elderly emerges. Section \ref{conclusion} discusses implications of results and concludes.

\section{Literature} \label{lit}

Economic analysis of aging has long focused on the connection between aging, savings behavior, and capital accumulation, central to many dynamic economic models. For example, the life-cycle model of \citet{modigliani66} and the overlapping generations (OLG) models of \citet{samuelson58} and \citet{diamond65} explicitly link these factors. The OLG model has proven particularly useful for studying how aging interacts with family structure, intergenerational transfers, and capital accumulation \citep{michel06}. Some models, such as those by \citet{boldrin02} and \citet{nishzhang95}, invert the traditional altruism paradigm, exploring child-to-parent altruism, which naturally generates elderly support and investment in children's education and fertility.

Other research emphasizes the strategic role of intergenerational transfers in family cohesion. For example, \citet{cigno93} shows that family-based transfers can be self-enforcing under specific economic conditions, such as the availability of capital markets and social support systems \citep{kaplanrobson03}. Studies like \citet{morand99} and \citet{blackburn05} further explore the connection between fertility, growth, and elderly support, highlighting transitions from high-fertility, traditional economies to low-fertility, modern ones, with changing patterns of intergenerational transfers.

Models of cultural transmission, such as those by \citet{bisver00, bisver01}, offer additional insights into elderly treatment. \citet{olivera13} demonstrates how norms of elderly care can persist despite declining fertility and rising life expectancy. Similarly, \citet{varvarigos21} links strong family ties, intergenerational transfers, and traditional means of production to evolving societal norms. \citet{cremer13} and \citet{ponthiere13} analyze the interaction between social support systems and familial altruism, showing how state-provided care can influence family dynamics and value transmission.

Our approach builds on these foundations but focuses on how cultural transmission interacts with elderly treatment. Unlike models centered on genetic or altruistic motives, we examine how middle-aged individuals cultivate their children's interest in cultural goods provided by the elderly, creating a direct link between cultural preservation and elderly support. This perspective aligns with \citet{cozzi98}, where cultural investment serves as a form of wealth storage for future consumption.

Empirical work on altruism toward the elderly supports the complexity of these dynamics. While early studies struggled to disentangle altruistic and self-interested motives \citep{laferrere99}, recent research, such as \citet{klima17}, finds evidence of `moderate altruism' driving intergenerational transfers.

Our model advances this literature by integrating strategic and altruistic aspects of elderly treatment with economic growth, property rights, and production. Unlike most prior research, which uses partial equilibrium approaches, we adopt a general equilibrium framework, as in \citet[Chapter 9]{acemoglu09}. This approach allows us to explore how preferences, technology, and elderly treatment co-evolve and to develop measures of `elderly treatment' and `elderly prestige.' It also connects technological change with the evolving status and treatment of the elderly, providing a comprehensive framework for understanding these dynamics across cultures and economies.

\section{The Ethnology of Elderly Treatment} \label{ethnosec}

\subsection{Cross-cultural research}

A foundational study in cross-societal analysis of elderly treatment is \citet{simmons70} (originally published in 1945), which surveyed 71 societies, including hunter-gatherers, pastoralists, and subsistence agriculturalists from diverse technological, geographical, and environmental conditions. Simmons aimed to capture a broad cross-section of societies to identify patterns in elderly treatment and how these patterns vary with subsistence strategies and societal complexity.\footnote{Simmons's methodology resembles the approaches used by G. P. Murdock and collaborators in the Ethnographic Atlas \citep{murdock69} and the Standard Cross-Cultural Sample \citep{murdockwhite69}, as well as data utilized in \citet{pryor73}. For applications of cross-cultural data in economic modeling, see \citet{bakerjacobsen2007a, bakerjacobsen2007b, bakermiceli2005}.}

Simmons’s work provides a comprehensive snapshot of elderly treatment across societies before the influence of global modernization and established a framework for subsequent cross-cultural studies. It highlights how the role of the elderly evolves with societal complexity and subsistence strategies. Our summary draws on Simmons' findings, the review by \citet{holmes76}, and the discussion in \citet{arnhoff64}.

\subsubsection{Food Sharing and Taboos} 

In societies without well-developed property rights, such as subsistence hunter-gatherer groups, communal food sharing and food taboos provide consumption security for the elderly. As property rights in land and livestock develop, elderly consumption becomes more secure \citep[Ch. 2]{simmons70}.  

Food sharing is common among hunter-gatherers, motivated by factors such as insurance, exchange, or resource management, though pinpointing a singular reason is difficult.\footnote{See \citet{bakerswope21} for an overview of sharing theories among hunter-gatherers, and \citet[Chapter 6]{kelly13} for an extensive discussion.} Sharing typically redistributes resources `vertically' from adults to both the elderly and children, and `horizontally' among hunters \citep{rlee03}. Moreover, hunter-gatherers often lack individual ownership rights over food and land, correlating shared food output with communal access to resources.\footnote{However, hunter-gatherers can still be territorial at the group level; see \citep{baker2003}.}  

Food taboos advantageous to the elderly, described as `Victual Deference' by \citet{silverman78}, are widespread among hunter-gatherers in various environments. Simmons noted that these taboos serve as ``a social expedient of particular advantage to the aged'' and are often manipulated by the elderly for personal benefit \citep[p. 26]{simmons70}. Examples from Simmons include the Inuit, Omaha, Pomo, Iban, and Arunta, where tabooed foods often favored the elderly. Additional examples include the \textit{Ju/'hoansi} of the Kalahari Desert, where elders are believed to have powers enabling them to consume foods considered too dangerous for younger individuals, such as ostrich eggs \citep[p. 33]{rosenberg09}. Similarly, among the Chukchee of Siberia, reindeer milk is deemed harmful to the young \citep{mcye81, bogoras04}. In the Cook Islands, the Pukapukan people maintained a system of food sharing and taboos that allowed elders to control distribution and reserve nutritious parts, such as belly fats and internal organs, for themselves \citep[p. 10]{simmons70}.  

\subsubsection{Property rights}

While nonexclusive food sharing benefits the elderly, Simmons observes that ``property rights have been lifesavers for the aged'' \citep[p. 36]{simmons70}, suggesting that well-defined property rights over assets also benefit them. Simmons distinguishes between these types of rights, emphasizing that while nonexclusive rights over food help sustain the elderly, ownership of assets like land or livestock provides additional advantages.  

Simmons documents asset accumulation by the elderly, such as working animals and livestock, in groups like the Navajo (Southwestern North America), Chukchi (Siberia), Yakut (Russian Far East), and Akamba (Kenya). In agricultural societies with land ownership, respect for the elderly often aligns with asset accumulation. As Simmons notes, ``Firmly entrenched property rights which enhanced the powers of the aged have also been very common among all cultivators of the soil'' \citep[p. 44]{simmons70}.

\subsubsection{Knowledge and the Elderly}

The elderly often secure livelihoods by providing valuable services, including administrative tasks that convey power and prestige. Among less developed societies in Simmons' sample, elders may receive payment for childcare, food processing, technical knowledge, or administrative expertise. Their status improves with the importance of tradition and the degree of specialized knowledge they possess.  

Simmons also highlights the provision of cultural goods by the elderly in exchange for compensation. For instance, elderly Navajo were paid for knowledge of magic songs and charms, while Pomo and Lengua medicine men earned fees for performing rites \citep[p. 37-38]{simmons70}. Cultural conditioning plays a role in establishing property rights over such intangible assets, as with rights over songs and charms, and in shaping practices like food tabooing or game-sharing rules that benefit the elderly.  

Beliefs in the magical powers of the elderly are widespread among subsistence hunters and gatherers, instilling fear and respect. Elders are often thought capable of bringing or curing disease, beliefs actively cultivated across varied contexts, including the Yahgan (Tierra del Fuego), Arunta (Australia), Greenland Inuit, Pomo (California Desert), and Witoto (Northern Peru). Among the Ju/'hoansi, some elders are believed to have healing powers derived from their ability to harness mystical forces without being overwhelmed \citep[p. 34]{rosenberg09}.

\subsubsection{Inculcation}

Attitudes and practices toward the elderly are consciously instilled in the young, often at significant cost to parents. \citet{simmons62} emphasizes that such treatment arises from deliberate application of resources and effort, stating, `Respect for old age has resulted from social discipline' \citep[p. 50]{simmons70}. Expanding on this point, Simmons observes:  

\begin{quote} ...while human beings can be as indifferent to the dependency needs of their ancestors...they *can also* learn, be taught, inspired, or impelled to respect, succor, and sustain their elders...these patterns of group and filial responsibility are ‘passed on’ from generation to generation...This learned behavior is called `culture'.” \citep[p. 37]{simmons62}  
\end{quote}

Cultural conditioning thus enables the elderly to access and control resources they might otherwise lack.

\subsection{Elderly Status and Cultural Complexity}

Subsequent cross-cultural research on elderly treatment reaches similar conclusions. \citet[p. 298]{press72} summarize \citet{cowgill71} and their ``ostensibly universal generalizations about the status of the aged,'' that align closely with \citet{simmons70}. These include the high status of the elderly in contexts where knowledge accumulation is economically important, property rights are accessible, and the elderly can perform meaningful economic tasks. In fact, \citet[p. 299]{press72} note that their independently derived determinants of elderly status coincided with Cowgill’s findings.  

Although much focus has been on hunter-gatherers and subsistence agricultural or pastoral societies, later research explores how elderly treatment evolves with increasing technological complexity. This involves changes in property rights, social complexity, and production methods. However, defining a consistent trajectory for elderly treatment as societies develop is challenging. Research suggests a nonlinear relationship, with treatment varying across levels of societal complexity. Terms like `status' and `good treatment' are often poorly defined, complicating analysis.\footnote{\citet{dowd81} and \citet{balkwell81} highlight the issues with vague terminology.} For example, comparisons of elderly treatment in hunter-gatherer versus modern industrial societies may hinge on unclear definitions of status and modernity.  

A significant characterization is the `curvilinear hypothesis' \citep{blumberg72}, as developed by \citet{nimkoff60}.\footnote{See \citet{lee96} for an overview, as well as \citet{lee81} and \citet{bakermiceli2005}.} This hypothesis posits a nonlinear relationship between social complexity and family structure: hunter-gatherers and advanced societies tend to have smaller families, while intermediate-level societies feature expansive familial structures. In societies of intermediate complexity, the elderly often hold high status, as documented by \citet{lee79}. However, this literature often conflates two factors: preferential treatment due to cultural norms and respect derived from accumulated resources. The core question remains: how do asset accumulation and cultural transmission interact?  

\subsection{Some quantitative evidence}

\citet{simmons70} provides extensive cross-tabulations of variables to highlight empirical regularities in elderly treatment. This section summarizes his data and presents evidence on the roles, inculcation, and treatment of the elderly across societies. Simmons’ data, while valuable, poses challenges for advanced analysis due to its small sample size and missing values. To address this, we constructed indices summarizing key societal features related to the elderly. These indices suggest meaningful relationships among societal characteristics. Appendix \ref{datapp} provides a detailed discussion of the data and index construction.  

The data includes qualitative traits scaled into indices, summarized in Table \ref{datadescription}. The first four indices consider the elderly’s role in directly economically productive activities, in social production, in advice and knowledge provision, and the holding of titular roles in the society. The fifth index combines the first three forms of production. The sixth index considers the openness of the society to market-based economic activity. The seventh index measures positive treatment of the elderly, and the eighth index measures the amount of positive inculcation towards the elderly.  

\begin{centering}
\begin{table}[h!]\centering
\begin{tabularx}{\textwidth}{X|X}
\toprule
Index    & Description           \\
\midrule
Contribution to direct production              & Tasks such as trading, childcare, hunting support, agriculture, household work, and collecting fees for services. \\
Contribution to social production              & Performing marriages, initiation rites, funerals, and roles as priests or shamans. \\
Knowledge and advice provision              & Providing advice and information. \\ 
Titular duties                          & Holding roles as chiefs, elders, and judges. \\
Combined production            & Sum of direct production, social production, and knowledge and advice provision indices. \\
Openness of economy       & Presence of property rights, trading, and currency usage. \\ \hline
Positive treatment             & Respect, distinction, and familial support for the elderly. \\
Positive inculcation             & Beliefs portraying the elderly as daimons, heroes, and/or children’s friends. \\
\bottomrule
\end{tabularx} 
\caption{Indices constructed from data in \citet{simmons70}.}
\label{datadescription}
\end{table}
\end{centering}

We then correlated the six production and openness of economy indices with the two indices for positive treatment of the elderly and positive inculcation towards the elderly. Table \ref{crosstabs} presents these correlations. Positive treatment of the elderly is positively correlated with the elderly’s contribution to production and with their provision of advice and knowledge, as well as with the combined production index and with their holding of titular duties. Positive inculcation regarding the elderly is also positively correlated with these categories except for titular duties. Positive treatment and positive inculcation are positively correlated. Interestingly, the open economy index is not correlated with  positive treatment and shows some evidence (though not at standard statistical significance) or negative correlation with positive inculcation. We shall consider these patterns in the context of consistency with our formal model later in the paper.   

\begin{table}[h!]
{
\def\sym#1{\ifmmode^{#1}\else\(^{#1}\)\fi}
\begin{tabular}{lcc}
\toprule
      &  Positive treatment & Positive inculcation           \\
\midrule
Contribution to direct production              & $0.38^*$    &  $0.24^*$    \\
Contribution to social production              &  0.21    &  $0.22$    \\
Knowledge and advice provision              &  $0.43^*$    & $0.45^*$     \\
Combined production             &  $0.41^*$    &   $0.33^*$   \\
Titular duties                           & $0.36^*$     & $0.15$   \\
Openness of economy       &  0.05    & -0.16     \\
Positive treatment             &   -   & $0.45^*$ \\
\bottomrule
\end{tabular} 
\caption{Pairwise correlations between indices based on \citet{simmons70}. $^*$ denotes significance at 95\%. See Appendix \ref{datapp} for variable details.}
\label{crosstabs}
}
\end{table}

\subsection{Modernization and the Elderly}

As globalization and technological change advance, population growth slows, and the global population ages, understanding the interaction between technical change and cultural practices related to the elderly becomes increasingly critical. While industrialization is often considered detrimental to elderly well-being \citep[e.g.,]{dowd81}, the relationship between industrialization and elderly welfare is nuanced. This complexity remains relevant even in modern societies, raising questions about how institutions like social security systems influence care for the elderly. Such questions grow in significance as many regions transition from traditional to modern economies in the coming decades.

For example, \citet{koda18} find that public transfers can crowd out private care for the elderly, and that the implications of modernization, including lower fertility and smaller families, tend to reduce familial support. These findings align with theories in \citet{koda18} and \citet{cigno93}. However, \citet{ko21}, in their study of pensions in rural China, report that public pensions increase both well-being and intergenerational exchange. Similarly, \citet{kodate17} describe the complex interplay between public funding, paid healthcare workers, cultural norms, and family care.

The treatment of the elderly can influence and be influenced by consumption smoothing. Allocating resources for the future may enhance elderly care through cultural norms, asset accumulation, or both, depending on property rights and the nature of production.

Our model is constructed to shed light on these trade-offs, offering insight into how elderly treatment may evolve in the coming decades. By capturing the interaction between elderly treatment and economic attributes, the model generates predictions consistent with the complexity observed in the ethnographic record. In addition, it helps clarify vague concepts in the empirical literature such as `status,' `prestige,' or what it means for the elderly to be `treated well.' 

\section{Theoretical Framework} \label{theorysec}

The model is a standard infinite-horizon overlapping-generations framework, where young, middle-aged, and elderly agents coexist at any given time. Decision-making occurs solely in middle age. Young agents derive no utility but may be inculcated with an appreciation for culture. In the base model, middle-aged individuals cultivate a preference in the young for giving gifts to the elderly. Upon reaching middle age, these individuals transfer resources to the elderly, deriving `warm glow utility' as described by \citet{andreoni89}. This willingness to transfer depends on the degree of appreciation instilled in them during youth by the current middle-aged cohort.\footnote{This approach resembles models in which cultural heritage and societal devotion are endogenously determined, such as \citet{cozzi98}. However, Cozzi's model assumes self-fulfilling beliefs about culture without active cultivation.}  

Middle-aged and elderly consumption at time $ t $ are denoted \( c_{mt} \) and \( c_{et} \), with gifts to the elderly represented by \( g_{mt} \). The lifetime utility of a middle-aged agent is given by the log-linear utility function:
\begin{equation} \label{lifetimeU}
    U_t = (1-\eta_t)\ln c_{mt} + \eta_t \ln g_{mt} + \beta \ln c_{et+1},
\end{equation}
where \( \eta_t \) measures the relative importance of gift-giving in middle-aged utility.

Middle-aged agents receive income \( y_{mt} \), allocated to consumption \( c_{mt} \), savings \( s_{mt} \), gifts \( g_{mt} \), and inculcative investment in the young. Inculcation determines the future demand parameter \( \eta_{t+1} \), reflecting the young's appreciation for culture when they reach middle age.\footnote{This contrasts with \citet{olivera13} and \citet{varvarigos21}, where inculcation increases the likelihood of preference similarity across generations. These can be described as `Simpsonian' models, inspired by Homer Simpson's parenting philosophy: ``Kids are the best… You can teach them to hate the things you hate.'' (\textit{The Simpsons}, season 11, episode 233, 1999).} The cost of inculcation is \( d_t \eta_{t+1} \), with \( d_t = \delta y_{mt} \) (\( \delta \in (0,1) \)), scaling costs to income.\footnote{This assumption ensures equilibria remain stationary, facilitating balanced-growth path analysis.} Middle-aged consumption is thus:
\begin{equation} \label{mac}
    c_{mt} = y_{mt} - d_t \eta_{t+1} - s_{mt} - g_{mt}.
\end{equation}

Elderly consumption \( c_{et} \) is derived from three sources: returns on previous savings \( R_t s_{mt-1} \), earned income \( y_{et} \), and received gifts \( g_{et} \). Income \( y_{et} \) may include compensation for productive or advisory roles, government transfers such as social security, or traditional sharing of output in hunter-gatherer societies. Elderly consumption is given by:
\begin{equation} \label{oac}
    c_{et} = R_t s_{mt-1} + y_{et} + g_{et}.
\end{equation}

This model of intergenerational gift-giving also captures cultural transmission. We will demonstrate that the utility function in (\ref{lifetimeU}) aligns with a model where cultural goods of endogenously determined value are exchanged across generations.

\subsection{Simple model mechanics}

We begin with a simplified version of the model in which there is no savings ($s_{mt}=0$ in Equation (\ref{oac})). This simplification serves three purposes: it clarifies the intuition by focusing solely on the inculcation and gift decisions of current middle-aged agents; it empirically aligns with premodern economies where labor and natural resources, rather than capital accumulation, drive output—a characteristic of hunter-gatherer and subsistence agricultural societies; and it provides insight into the interplay between property rights, inculcation, and elderly treatment.

Substituting middle-aged consumption (\ref{mac}) and elderly consumption (\ref{oac}) into (\ref{lifetimeU}) and setting $s_{mt}=0$, we obtain:

\begin{equation} \label{nocapira}
U_{mt} = (1-\eta_t)\ln (y_{mt}-g_{mt}-d_t)+\eta_t \ln g_{mt} +\beta \ln \left(y_{et+1}+g_{et+1}\right)
\end{equation}

To simplify further, we assume $\eta_t \in {0,\eta}$ with $0<\eta<1$, so the current middle-aged generation makes a binary choice: whether to inculcate children to enjoy giving gifts to the elderly at intensity $\eta$ and cost $d_t$, or not. Simultaneously, they decide on a gift for the current elderly, based on the previous generation's inculcation. These assumptions allow us to model the problem as a stationary, infinite-horizon game where each generation sequentially decides on inculcation, anticipating that future generations will adhere to the subgame-perfect Nash equilibrium strategy.

Middle-aged utility in an inculcation equilibrium is:

\begin{equation} \label{nocapir}
U_{mt, \eta} = (1-\eta)\ln (y_{mt}-g_{mt}-d_t)+\eta \ln g_{mt} +\beta \ln \left(y_{et+1}+g_{et+1}\right)
\end{equation}

Here, $g_{et+1}$ is the expected gift in the next period, given that current youth are inculcated. Stability of this equilibrium can be analyzed by considering a deviation where a middle-aged agent opts out of inculcation, leading to no future gifts. In this case, the agent's lifetime utility becomes:

\begin{equation} \label{nocapnir}
U_{mt, 0} = (1-\eta)\ln (y_{mt}-g'{mt})+\eta \ln g'{mt}+\beta \ln (y_{et+1})
\end{equation}

When inculcation and gift-giving are maintained, optimal gifts are derived by differentiating (\ref{nocapir}) with respect to $g_{mt}$, yielding:

\begin{equation} \label{optsimpg}
g_{mt}^*=\eta (y_{mt}-d_t)
\end{equation}

Thus, the middle-aged agent gives a gift proportional to disposable income, defined as income remaining after inculcation. Using this, equilibrium gifts received, $g^{et}$, can be determined. To account for demographic changes, we assume a constant population growth rate $n$, such that $n{t+1}=(1+n)n_t$. In equilibrium, this implies $g^{et}=(1+n)g^*{mt}$.

If the middle-aged agent abandons inculcation, the final gift given is $g'{mt}=\eta y{mt}$, and no gift is received during old age in the subsequent period.

\subsubsection{Gifts and transfer of cultural goods}

Before continuing to analyze inculcative decisions, we show that the warm glow gift-giving model is equivalent to a model in which cultural goods are transferred between generations. This allows us to relate the importance of gifting in support of the elderly to a more general notion of cultural transmission and its importance in maintaining the elderly. 

Suppose that instead of enjoying giving gifts, the current middle-aged generation places value on consumption of a cultural good, for which a price $p_t$ must be paid. The demand for cultural goods is $x^d_{t}$ and the good is supplied by current elderly, each of whom has one unit of the good to supply. Lifetime utility (\ref{lifetimeU}) in these circumstances is:

\begin{eqnarray} \label{lifetimeU-dem}
    U_t =(1-\eta_t)\ln{c}_{mt} +  \eta_t \ln x^d_{t}+ \beta\ln c_{et+1}
\end{eqnarray}

Setting $s_{mt}=0$ and modifying consumption of middle-aged and elderly (\ref{mac}) and (\ref{oac}) accordingly gives (\ref{lifetimeU-dem}) as:

\begin{eqnarray} \label{lifetimeU-dem-rf}
    U_t =(1-\eta)\ln (y_{mt}-g_{mt}-p_tx_{dt})+\eta \ln x_{dt} +\beta \ln \left(y_{et+1}+x_{st+1}p_{t+1}\right)
\end{eqnarray}

In specification (\ref{lifetimeU-dem-rf}), $\eta$ now measures the relative importance of cultural consumption to material consumption. Current middle-aged agents spend $p_tx_{mt}$ on cultural goods, which gives them utility $\eta \ln x_{mt}$.   

The current middle-aged inculcate an interest in cultural goods in the current young. When cultural goods are more valuable to the young, as reflected in a higher value of $\eta_{t+1}$, they pay more for them, enriching the elderly.\footnote{This is similar to a model in which firms engage in advertising spending which changes the demand for a good as described in \citet{bagwell07}, for example.} We can now establish that this cultural goods model is observationally equivalent to the warm-glow gift-giving model under our assumptions. This leads us to Proposition \ref{equivalence}.

\begin{proposition}[Equivalence of gift-giving and cultural exchange]\label{equivalence} The models captured by utility functions (\ref{lifetimeU}) and (\ref{lifetimeU-dem}) are observationally equivalent:
\begin{itemize}
    \item \textbf{Gift-giving:} Middle-aged agents choose gifts to the elderly, and receive `warm-glow' utility from giving gifts.    \item \textbf{Cultural exchange:} Middle-aged agents spend $x^d_tp_t$ on cultural goods supplied by the elderly and earn utility from these goods.
\end{itemize} 

\end{proposition} 
\begin{proof}[Proof:]
The gift-giving model resulted in an equilibrium gift given by current middle-aged agents proportional to current middle-aged disposable income, as in (\ref{optsimpg}). For the cultural goods model, differentiating (\ref{lifetimeU-dem-rf}) with respect to $x^d_t$ and solving gives the demand for $x^d_t$ as:

\begin{equation} \label{cultdemand}
    x_t^d=\eta\frac{y_{mt}-d_t}{p_t}
\end{equation}

Each elderly agent has a supply of one unit of cultural goods, so demand equal to supply requires $x_t^s=(1+n)x_t^d$ or $1=(1+n)x_t^d$. Equating demand from (\ref{cultdemand}) with supply and solving gives $p_t=(1+n)\eta(y_{mt}-d_t)$. It follows that $p_tx_t^d=\eta(y_{mt}-d_t)$ are the expenditures of middle-aged agents on cultural goods, while $p_{t+1}x^s_{t+1}=p_{t+1}=(1+n)\eta(y_{mt+1}-d_{t+1})$ is the amount received by the elderly, identical to the gifts implied by $g^*_{mt}$ in (\ref{optsimpg}) and the resulting gifts received by the elderly, $g_{et}^*$.
\end{proof}

The equivalence result in Proposition \ref{incymye} is useful in that it allows us to work with the easier-to-solve gift-giving model while also describing how cultural transmission interacts with elderly treatment. 

\subsubsection{Inculcation decisions}

We can now show that the inculcation decision depends on whether or not elderly income is small relative to middle-aged income. This result and the key income threshold are described in Proposition \ref{incymye}.

\begin{proposition}[Relative income and inculcation]\label{incymye} Inculcation occurs when elderly income is sufficiently small relative to middle-aged income, occurring when:
\begin{equation*}
\frac{y_{et+1}}{y_{mt+1}} \leq Y^*
\end{equation*}
where
\begin{equation*}
Y^*=\frac{(1+n)\eta(1-\delta)}{\left((1-\delta)^{-\frac{1}{\beta}}-1\right)}
\end{equation*}
\end{proposition} 
\begin{proof}[Proof:]
Substitution of the gift in (\ref{optsimpg}) into (\ref{nocapir}), and substitution of the terminal gift $g'_{mt}=\eta y_{mt}$ into (\ref{nocapnir}) coupled with the observation that $g_{et+1}=(1+n)g_{mt+1}$ allows writing the difference in utilities along the equilibrium path with inculcation and the utility from discontinuing inculcation, $\Delta U_t=U_{mt,\eta}-U_{mt,0}$ as:
\begin{equation} \label{deltaexp} 
\Delta U_t=\ln (y_{mt}-d_t) - \ln y_{mt} + \beta \ln (y_{et+1} + (1+n)\eta (y_{mt+1}-d_{t+1})) - \beta \ln ( y_{et+1}) 
\end{equation}

$\Delta U_t$ in (\ref{deltaexp}) simplifies to:

\begin{equation} \label{deltaexp2} 
\Delta U_t=\ln\left(1-\frac{d_t}{y_{mt}}\right) + \beta \ln\left( 1 + (1+n)\frac{\eta (y_{mt+1}-d_{t+1})}{y_{et+1}}\right) 
\end{equation}
Using the assumption that $d_t$ is proportional to $y_{mt}$, so $d_t=\delta y_{mt}$, where $\delta \in (0,1)$ is a fixed positive constant, and substituting into (\ref{deltaexp2}) gives:
\begin{equation} \label{simpineq}
\Delta U_t \geq 0 \quad \textrm{if} \quad    \frac{y_{et+1}}{y_{mt+1}} \leq \frac{(1+n)\eta(1-\delta)}{\left((1-\delta)^{-\frac{1}{\beta}}-1\right)}=Y^*
\end{equation}

\end{proof}

Proposition \ref{incymye} defines a critical threshold for the \textit{relative income} earned by the elderly. If relative elderly income falls below the threshold $Y^*$, inculcation occurs in equilibrium. The threshold increases (making inculcation more likely) if population growth increases ($n \uparrow$); this is because there are more middle-aged agents for each elderly agent, resulting in more gifts received. If the effectiveness of inculcation increases ($\eta \uparrow$), if inculcation becomes cheaper ($\delta \downarrow$), and if old-age consumption becomes more important ($\beta \uparrow$), inculcation becomes more likely. When inequality (\ref{simpineq}) does not hold, middle-agers do not anticipate the need to inculcate current youth, as the income they earn when old generates sufficient income for sustenance. Also, the relatively low income of middle-agers does not produce sufficiently large gifts to warrant inculcation.  

We now add to the model a measure of the status or prestige of the elderly: the relative consumption of the elderly, $\frac{c_{et}}{c_{mt}}$. Using (\ref{optsimpg}) and (\ref{simpineq}), we write this consumption ratio as follows, accounting for the range of relative incomes for which inculcation and concurrent gift-giving occurs: 

\begin{equation} \label{Crat}
\begin{aligned}
    \frac{c_{et}}{c_{mt}}&=\frac{y_{et}}{y_{mt}}, &\frac{y_{et}}{y_{mt}} >Y^* \\
    &=\frac{y_{et}}{y_{mt}(1-\eta)(1-\delta)} + \frac{(1+n)\eta}{1-\eta}, & \frac{y_{et}}{y_{mt}} \leq Y^*
\end{aligned}
\end{equation}

Equation (\ref{Crat}) suggests that there may not be a direct link between inculcative activity and elderly treatment. The elderly may provide services or other productive labor so that $y_{et}$ is large relative to $y_{mt}$. The elderly will then be well treated, but this could be a result of the nature of production, not of inculcation. 

The ethnographic record highlights production and property rights as key factors shaping the elderly's economic roles, prompting us to extend the model accordingly.

\subsection{Production Technology} \label{ge}

Suppose that in each time period, society creates output using labor $L_t$ and resources $T$ that do not vary with time. We will refer to $T$ as land, but it could represent some other resource used to produce output. Output is produced via the constant-returns production function:

\begin{equation} \label{production_function}
    Y_t=F(L_t,T)
\end{equation}

Let $A_{it},i=m,e$ denote the effective labor endowments of middle-aged and elderly agents. Agents supply their labor inelastically, so total labor in (\ref{production_function}) is effective labor supplies multiplied by population levels:

\begin{equation} \label{els}
L_t=N_{mt}A_{mt}+N_{et}A_{et}, 
\end{equation}

The production function (\ref{production_function}) is constant-returns so it can be decomposed using Euler's theorem. Using $F_L=\frac{\partial F}{\partial L}$ and $F_T=\frac{\partial F}{\partial T}$, we may write:

\begin{equation} \label{eulerdecomp}
    Y_t=F(L_t,T)=F_LL_t+F_TT
\end{equation}

The decomposition in (\ref{eulerdecomp}) gives us a natural way to attribute output to factors of production; $F_LL$ is labor's contribution to output, while $F_TT$ is output created by fixed inputs. Similarly, the contribution of any agent's labor to output is $F_Ll_i$, so output created by the labor of different ages of agents is:

\begin{equation} \label{labinc}
    q_{mt}^L = A_{mt}F_L,\quad q_{et}^L = A_{et}F_L
\end{equation}

output attributable to the fixed resource is 
\begin{equation} \label{landincome}
    q_{t}^T = TF_T
\end{equation}

Property rights define how the output created in (\ref{labinc}) and (\ref{landincome}) is allocated to agents.  

\subsection{Property Rights}

In Section \ref{ethnosec}, we discussed Simmons' intriguing remarks on the role of property rights in supporting the elderly. While he describes property rights as `lifesavers' for the elderly \citep[p. 36]{simmons70}, he also notes that weaker property rights can sometimes benefit them, such as when food is shared communally. To clarify, we distinguish between property rights over inputs like land and property rights over output such as game acquired. Simmons appears to blur these distinctions in his ethnographic overview. Strong property rights over productive assets likely benefit the elderly by ensuring resources are invested in ownership or accumulation. Conversely, weak property rights over output, encouraging communal sharing, also seem advantageous for them. This input-output distinction aligns naturally with the Euler decomposition of the production function (\ref{eulerdecomp}).

Suppose that each agent has an ownership share of the fixed resource $\sigma_{it},i=m,e$ where $\sum \sigma_{it}=1$ across the population or, in our case, $\sigma_{et}N_{et}+\sigma_{mt}N_{mt}=1$. Then, a starting point to describe the output that an agent is responsible for creating is expression (\ref{zeroprop}):

\begin{equation} \label{zeroprop}
    q_{it}=q_{it}^L+q_{it}^T=A_{it}F_L+\sigma_{it} F_TT, \quad i=m,e; \quad \sum_N \sigma_{it}=1
\end{equation}

\subsubsection{Property rights over land}

The specification in (\ref{zeroprop}) naturally frames property rights over land (or other assets) in terms of $\sigma$. Consider a situation in which ownership of land is not defined at all, that is, land is a free-access resource. $\sigma_{it}$ should then be set so as to allow our model to mimic a typical common production model, in which agent $i$'s output depends upon $i$'s relative share of labor, so $q_{it}=q_{it}^L=\frac{l_i}{L}F(L,T)$.\footnote{This is a slight generalization of the idea that open access resources generate average product as returns for users as in, for example, \citet{gordon54}. See \citet{bakerconning24} for further discussion.}

Open-access land rights is then captured by $\sigma_{it}$ as follows:
\begin{equation} \label{freeaccesssig}
    \sigma^0_{it}=\frac{A_{it}}{L_t}=\frac{A_{it}}{N_{mt}A_{mt}+N_{et}A_{et}},\quad i=m,e
\end{equation}

To confirm that this is in accordance with a typical common property, open access model, note that we can combine agents' output (\ref{zeroprop}) with $\sigma_{it}^0$ in (\ref{freeaccesssig}) to get:
\begin{equation} \label{comproof}
q_{it} = A_{it}F_L + \frac{A_{it}}{L_t}F_tT = \frac{A_{it}}{L_t}\left(F_LL_t +F_tT\right)=\frac{A_{it}}{L_t}F(L_t,T)
\end{equation}

Where the last part of (\ref{comproof}) follows from the assumption that $F(L_t,T)$ is homogeneous of degree one. This pins down one side of a spectrum of possible degrees of ownership rights over land, with $\sigma_{it}^0$ in (\ref{freeaccesssig}) corresponding to a situation in which there are no land property rights.

When property rights are fully defined over land, suppose that a right of first possession determines ownership. This can be approximated by the elderly controlling assets in a way proportional to past usage (that is, their usage during middle age), so that when asset rights are fully defined, $\sigma_{et}=\frac{A_{mt-1}}{A_{mt-1}N_{mt-1}}=\frac{1}{N_{et}}$. It follows that in this case the middle-aged do not have any property rights, so $\sigma_{mt}=0$. Fully defined property rights over land are then characterized by:
\begin{equation} \label{propsig}
    \sigma^1_{et} = \frac{1}{N_{et}},\quad \sigma_{mt}^1=0
\end{equation}

In (\ref{freeaccesssig}) and (\ref{propsig}), we have described free access and fully defined rights over land using values for $\sigma_{it}$ as $\sigma_{it}^0$ and $\sigma_{it}^1$. A spectrum of intermediate possibilities exists between the two endpoint cases, and we introduce a parameter $\phi_T \in [0,1]$ to capture the degree of definition of property rights over land (or other fixed assets). This results in values the following values for $\sigma$: 

\begin{equation} \label{elderlyland}
    \sigma_{et} = \phi_T\frac{1}{N_{et}}+(1-\phi_T)\frac{A_{et}}{L_t}
\end{equation}

\begin{equation} \label{middleland}
    \sigma_{mt} = (1-\phi_T)\frac{A_{mt}}{L_t}
\end{equation}

When $\phi_T=0$ in (\ref{elderlyland}) and (\ref{middleland}) there are no ownership rights in land, while when $\phi_T=1$ property rights are completely defined.  

\subsubsection{Property rights over output}
We model imperfect rights over output as the loss of some share of $q$ to a general output fund, which is then shared equally among all agents. Let $\phi$ denote the share of own output retained by the agent so that $1-\phi$ of each agent's output is contributed to the output pool. Then we can write the total income received by an agent as a function of the degree of definition of property rights as captured by different values of $\phi$ and $\phi_T$, assuming that each agent gets an equal share of the total output shared with others:
\begin{equation} \label{outshare}
    y_{it}= \phi q_{it} + (1-\phi)\frac{\sum_{j=1}^{N_t}q_{jt}}{N_t} \quad i=m,e
\end{equation}
Substituting into (\ref{outshare}) our expressions for outputs in (\ref{zeroprop}) gives:
\begin{equation} \label{agentincome}
    y_{it} = \phi \left(A_{it}F_L + \sigma_{it}F_TT\right)+(1-\phi)\frac{F(L_t,T)}{N_t} \quad i=m,e
\end{equation}

In (\ref{agentincome}) we see that imperfect ownership of output removes a bit of output from each agent but entitles an agent to a share of average output. 

\subsubsection{Property rights in general}

A stylized fact is that property rights tend to develop in tandem;  ownership of resources and the output created by those resources seem to evolve in the same direction.\footnote{Developing a theory of the development of property rights is beyond the scope of this paper. For a recent review and one modeling approach, see \citet{bakerconning24}. In particular see \citet{baker2008} for a model in which rights over assets emerge endogenously, and \citet{bakerswope21} for a model in which ownership of resources and output-sharing are interrelated.} We model this fact by assuming that $\phi_T$ and $\phi$ are linked. To allow for different kinds of property rights to develop at different rates yet still be linked, we introduce $\rho$ as a parameter characterizing the relative speed at which ownership of output and inputs co-evolve and then specify that $\phi_T=\phi^\rho$. When $\rho>1$, $\phi_T=\phi^\rho<\phi$, and rights over output are more developed than rights over assets for any interior value of $\phi$. Alternatively, if $\rho<1$, the opposite is true: rights over assets are more developed than rights over output, as $\phi_T>\phi$. 

\subsubsection{Back to relative treatment of the elderly}
We now can relate the development of property rights to inculcation, treatment of the elderly, and cultural transmission. Substituting (\ref{elderlyland}) and (\ref{middleland}) into (\ref{agentincome}), and forming the ratio $\frac{y_{et}}{y_{mt}}$ gives:

\begin{equation} \label{incratfull}
   \frac{y_{et}}{y_{mt}} = \frac{\phi\left(A_{et}F_L+\left[ \frac{A_{et}(1-\phi^\rho)}{L_t}+\frac{\phi^\rho}{N_{et}}\right]F_TT\right)+(1-\phi)\frac{F(L,T)}{N_{et}+N_{mt}}}{\phi \left(A_{mt}F_L+\frac{A_{mt}(1-\phi^\rho)}{L_t}F_TT\right)+(1-\phi)\frac{F(L,T)}{N_{et}+N_{mt}}}
\end{equation}

Assume that the aggregate production function is Cobb-Douglas: $F(L,T)=L^\alpha T^{1-\alpha}$. Then, expression (\ref{incratfull}) becomes:

\begin{equation} \label{incratcd}
   \frac{y_{et}}{y_{mt}} = \frac{\phi\left(A_{et}\alpha + \left[ A_{et}(1-\phi^\rho)+\frac{\phi^\rho L_t}{N_{et}}\right](1-\alpha) \right)+(1-\phi)\frac{L_t}{N_{et}+N_{mt}}}{\phi \left(A_{mt}\alpha + A_{mt}(1-\phi^\rho)(1-\alpha) \right)+(1-\phi)\frac{L_t}{N_{et}+N_{mt}}}
\end{equation}

Expression (\ref{incratcd}) links relative elderly income, technology of production, and property rights. It provides direct insight into how well the elderly are treated given the fundamental characteristics of the economy on its own, but it also functions as the driving variable in whether or not we expect to see inculcation and a cultural factor in the treatment of the elderly as it appears in the inequality (\ref{simpineq}) in Proposition \ref{incymye}. When $\phi=0$, (\ref{incratcd}) collapses to unity. This is because when all output is shared equally, elderly and middle-aged are the same. In Proposition \ref{relincpr}, we collect the key results concerning relative income and the definition of property rights embodied in (\ref{incratcd}).

\begin{proposition}[Property rights and the relative income of the elderly]\label{relincpr} The relative income of the elderly, $\frac{y_{et}}{y_mt}$ is U-shaped in the degree of definition of property rights $\phi$, provided a) production is not too labor-intensive ($\alpha$ is not too large), and b) provided $A_e<A_m$; the elderly have a relatively smaller endowment of labor.  
\end{proposition}
\begin{proof}[Proof:]
In what follows, the total population is $N_{et}+N_{mt}=N_t$. Taking the derivative of relative income $\frac{y_{et}}{y_{mt}}$ in (\ref{incratcd}) and simplifying gives:
\begin{equation} \label{dincratsimp}
    \frac{\partial \left(\frac{y_{et}}{y_{mt}}\right)}{\partial \phi}=\frac{N_tL_t\left(A_{et}N_{et} +A_{mt}(N_t\phi^\rho(1-\alpha)(1+\rho)-N_{et})\right)}{N_{et}\left((1-\alpha)A_{mt}N_t\phi^{\rho+1}-(A_{mt}-A_{et})N_{et}\phi-L_t\right)^2}
\end{equation}
Evaluating (\ref{dincratsimp}) at $\phi=0$ gives:
\begin{equation} \label{incratp0}
\frac{\partial \left(\frac{y_{et}}{y_{mt}}\right)}{\partial \phi}= (A_{et}-A_{mt}) \frac{N_t}{L_t}
\end{equation}
Which is negative so long as $A_{et}<A_{mt}$. Therefore, the relative income ratio starts at $1$ and initially decreases in $\phi$. The sign of the derivative in (\ref{dincratsimp}) thereafter depends wholly on the expression in the numerator of (\ref{dincratsimp}):
\begin{equation} \label{numerexp}
    A_{et}N_{et} +A_{mt}(N_t\phi^\rho(1-\alpha)(1+\rho)-N_{et})
\end{equation}

This expression is increasing in $\phi$ and equals zero at a critical value of $\phi$, $\phi^*$, given by:
\begin{equation}
\phi^*=\left(\frac{(A_{mt}-A_{et})N_t}{A_{mt}N_t(1-\alpha)(1+\rho)}\right)^\frac{1}{\rho}
\end{equation}
As the derivative in equation (\ref{incratp0}) is initially negative, equals zero at a unique point $\phi^* \in (0,1)$ and then becomes positive, relative income is U-shaped in the degree to which property rights are developed, $\phi$.
\end{proof}

Proposition \ref{relincpr} starts with the assumption that $A_{et}<A_{mt}$, which means that an elderly agent provides fewer effective units of labor than a middle-aged agent. Another assumption is that $\alpha$ is not very large; if it were $\phi^*$ would be outside the $(0,1)$ range and relative elderly income would only decreases in $\phi$. When $\phi=1$ we have:

\begin{equation} \label{incratp1}
\frac{y_{et}}{y_{mt}}=\frac{1}{\alpha}\frac{A_{et}}{A_{mt}}+\frac{1-\alpha}{\alpha}\frac{N_{mt}}{N_{et}}
\end{equation}

The limit ratio in (\ref{incratp1}) can be greater than one or less than one, depending on the relative elderly population and their relative productivity. For small enough $\alpha$ (that is, production is resource intensive), the ratio in (\ref{incratp1}) can be greater than one even if $A_{et}=0$ (the elderly do not supply any material labor). Then, the elderly can achieve a high level of prestige measured by relative income even in the absence of inculcation. This is because the elderly have accumulated assets in middle age. 

\begin{figure}[htbp] 
\begin{center}
\includegraphics[width=.7\textwidth]{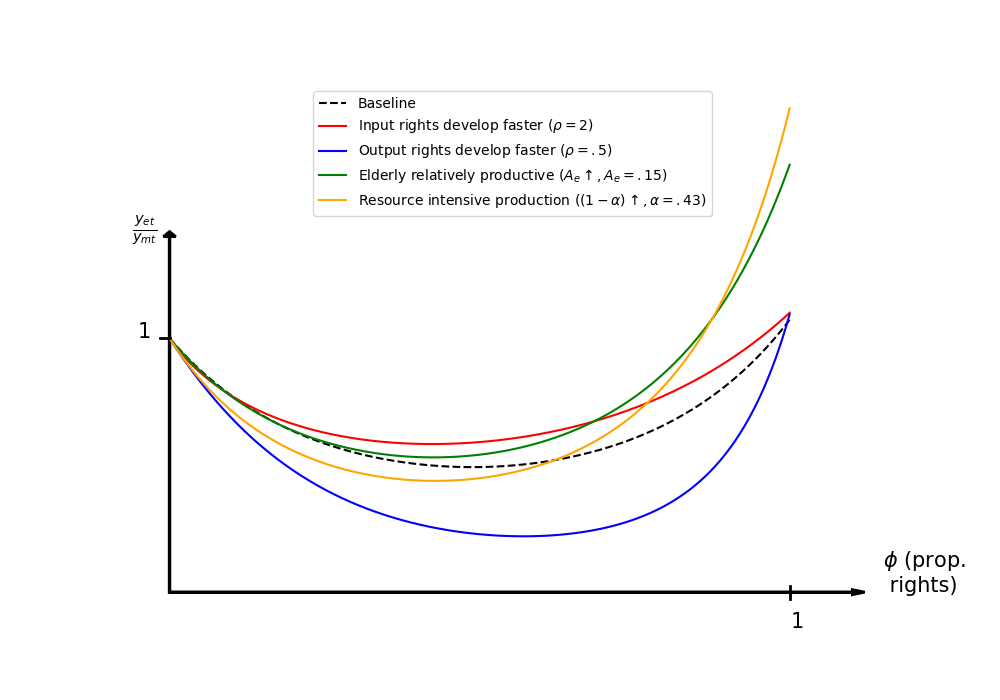}
\caption{Relative consumption of the elderly ($\frac{c_{et}}{c_{mt}}$), as functions of the development of property rights in output and productive resources from nonexistent ($\phi=0$) to complete ($\phi=1$). Baseline parameter values of $A_m=1$, $A_e=0.025$, and $\alpha=0.5$.} \label{samplerats}
\end{center} 
\end{figure}

\begin{figure}[htbp] 
\begin{center}
\includegraphics[width=.7\textwidth]{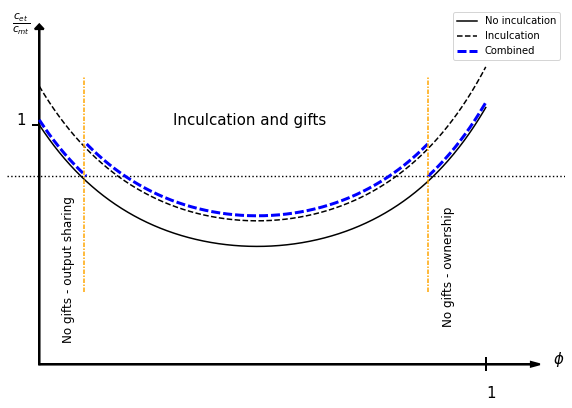}
\caption{Relative consumption of the elderly ($\frac{c_{et}}{c_{mt}}$), as a function of the degree to which property rights are developed, $\phi$, with endogenous inculcation. Baseline parameter values of $A_m=1$, $A_e=0.025$, and $\alpha=0.5$.} \label{consumrat}
\end{center} 
\end{figure}
Figure \ref{samplerats} plots the relative income ratio (\ref{incratcd}) as a function of property rights development, $\phi$, under varying assumptions about production. While model parameters can shift the curve's position, its overall shape remains consistent. The figure aligns with the empirical evidence discussed earlier. It suggests that in relatively egalitarian societies, cultural protections for the elderly may be unnecessary. However, as property rights initially develop, downward pressure on the relative consumption of the elderly creates demand for mechanisms like cultural inculcation to ensure their livelihood. For example, cultural goods or food taboos can transfer current income into future consumption. As resource ownership expands, asset income begins to benefit the elderly directly, reducing the need for such institutions, even as their prestige continues to rise.

Figure \ref{consumrat} combines results from Propositions \ref{incymye} and \ref{relincpr}. As property rights develop ($\phi$ increases) and the economy moves rightward along the $x$-axis, the relative consumption ratio $\frac{c_{et}}{c_{mt}}$ (equivalent to $\frac{y_{et}}{y_{mt}}$ in this range) initially declines, crossing the $Y^*$ threshold from Proposition \ref{incymye} at the vertical dashed line. At this point, inculcation begins, encouraging young people to support the elderly, mitigating the decline in relative treatment. However, once property rights become sufficiently advanced, inculcation ceases, as shown in the figure. This analysis highlights the complexity of linking culture, inculcation, and elderly prestige, suggesting a nonlinear relationship between culture, production, and elderly treatment.

\subsection{Savings, Capital Accumulation, and Inculcation}

We now analyze a version of the model where assets are accumulated through conscious savings decisions, exploring how this interacts with inculcation and gifts to the elderly. This approach also connects inculcative investment with savings and aligns the exposition with classical growth models featuring technological progress. 

The optimization problem of a current middle-aged agent can be written as the problem of maximizing $U_t$ in (\ref{lifetimeU}) subject to constraints (\ref{mac}) and (\ref{oac}), as follows:

\begin{equation} \label{maxprob}
\begin{aligned}
\max_{g_t,s_{mt},\eta_{t+1}} \quad & (1-\eta_t)\ln{c}_{mt} +  \eta_t \ln g_{mt}+ \beta\ln c_{et+1} \\
\textrm{s.t.} \quad & \quad c_{mt}=y_{mt}-d_t\eta_{t+1}-s_{mt}-g_{mt},  \\ 
& \quad c_{et+1} = s_{mt}R_{t+1}+(1+n)g_{et+1}+y_{et+1} \\
\end{aligned}
\end{equation}

The solution of (\ref{maxprob}) presents two challenges. First, a solution must take into account that current middle-aged agents' future gifts depend on current inculcative investment through the parameter $\eta_{t+1}$, and in an equilibrium solution to the model, beliefs about this dependence must be correct. Second, a solution must admit the possibility of solutions in which $\eta^*_t=0$. The subgame-perfect Nash equilibrium solution to the problem (\ref{maxprob}) has a structure in which agents give gifts in proportion to the present discounted value of lifetime income, with this proportion depending on $\eta_t$. For the time being, we have bypassed the issue of the determination of $\eta_t^*$ and skip to the form of the gift-giving function, which is given in Proposition \ref{giftfunny}:

\begin{proposition}[Gifts with accumulation]\label{giftfunny} If $\eta_t>0$ in equilibrium, gifts $g^*_{mt}$ are given by:
\begin{equation} \label{reactfun}
    g^*_{mt}=\frac{\eta_t}{1+\beta}\left(y_{mt}+\frac{y_{et+1}}{R_{t+1}}\right)
\end{equation}
\end{proposition} 

\begin{proof}[Proof:] See Appendix \ref{appa}.
\end{proof}

Proposition \ref{giftfunny} shows how equilibrium gifts and therefore, by analogy to Proposition \ref{equivalence}, cultural transmission, are proportional to lifetime income. Using Proposition \ref{giftfunny}, we can then begin to solve the maximization problem in (\ref{maxprob}). Two aspects of this solution are highlighted in Proposition \ref{twofeatures}.

\begin{proposition}[Total savings and gross returns on savings]\label{twofeatures} If $\eta_t>0$ in equilibrium, middle-aged agents' aggregate savings - the sum of inculcative investments and conventional savings - is described by:
\begin{equation} \label{fulleq1}
s_{mt}+\eta_{t+1}d_t=\frac{\beta y_{mt}-\frac{y_{et+1}}{R_{t+1}}}{1+\beta}
\end{equation}
In equilibrium, the gross rate of return on conventional savings is given by:
\begin{equation} \label{fulleq2}
    R_{t}=\frac{1+n}{1+\beta}\frac{y_{mt}+\frac{y_{et+1}}{R_{t+1}}}{d_{t-1}}
\end{equation}
\end{proposition} 
\begin{proof}[Proof:] See Appendix \ref{appb}.
\end{proof}

Equation (\ref{fulleq1}) in Proposition \ref{twofeatures} describes the total amount middle-aged agents wish to allocate to future consumption. Expenditures on savings and inculcation add up to the fraction of income a middle-aged agent wishes to save, which depends upon current income less discounted future income. Equation (\ref{fulleq2}) in Proposition \ref{twofeatures} requires that the rates of return on the different means of transferring current wealth to the future in (\ref{fulleq1}) be equilibrated. The rate of return on inculcative investments and conventional savings must be equal; otherwise, middle-aged agents would not employ both savings instruments. As we maintain the assumption that $d_t$ is proportional to income ($d_{t}=\delta y_{mt}$), we can interpret (\ref{fulleq2}) to mean that returns to assets such as capital must be equal in equilibrium to the unit return on inculcative investment. 

The question remains whether an equilibrium with gifts and inculcation is sustainable; i.e., whether $\eta_t>0$ in equilibrium. We need to round out the model to see how $\eta_t$, $s_t$, and $R_t$ depend on the underlying nature of the economy, as they are quantities determined in equilibrium and require specification of the production side of the economy. 

\subsubsection{Production and Economic Growth}

We again use a Cobb-Douglas production technology in which production occurs according to the aggregate production function $Y_t=L_t^\alpha K _t^{1-\alpha}$, where $K_t$ now denotes the total capital stock and as before $L_t$ denotes total effective labor supply:
\begin{equation} \label{effls}
    L_t=A_{mt}N_{mt}+A_{et}N_{et}
\end{equation}
We assume that $A_{it},i=e,m$ grows at a constant exogenous rate $a$, while the population grows at a rate $n$. Constant population growth means that $N_{mt}=(1+n)N_{et}$, and that $\tfrac{1}{1+n}$ is a measure of the relative size of the elderly population. 

We use (\ref{effls}) and the production technology to obtain the following expressions for middle-aged labor income, elderly labor income, and the gross rate of return on capital, assuming that all resources are paid their marginal products, where $k_t=\frac{K_t}{L_t}$ is capital per effective unit of labor:

\begin{equation} \label{wagesrates}
\begin{aligned}
y_{mt} &=& \alpha A_{mt}\left(\frac{K_t}{L_t}\right)^{1-\alpha} &=&\alpha A_{mt}k_t^{1-\alpha} \\
y_{et} &=& \alpha A_{et}\left(\frac{K_t}{L_t}\right)^{1-\alpha} &=&\alpha A_{et}k_t^{1-\alpha} \\
R_t    &=& (1-\alpha)\left(\frac{K_t}{L_t}\right)^{-\alpha} &=&(1-\alpha)k_t^{-\alpha}
\end{aligned}
\end{equation}

We also define some new terms:
\begin{eqnarray}
\tau_{et}&=\frac{A_{et}}{A_{mt}}&=\textrm{relative effective labor endowment of an elderly agent} \nonumber \\
\gamma_{et}&=\frac{N_{et}A_{et}}{L_t}&=\textrm{relative aggregate labor supply of the elderly} \label{taugam}
\end{eqnarray}

Given our assumptions of a constant rate of technological progress, we can omit time subscripts for the terms introduced in (\ref{taugam}) and write $\tau_e$ and $\gamma_e$. One can also show that $\tau_e$ and $\gamma_e$ are related via $\gamma_e = \frac{\tau_e}{\tau_e+(1+a)}$. 

Application of factor returns in (\ref{wagesrates}) to (\ref{fulleq1}) and (\ref{fulleq2}) yields Proposition \ref{dynolaws}, which translates Proposition \ref{twofeatures} into a pair of dynamic equations that describe the evolution of the capital stock and inculcation:

\begin{proposition}[Dynamics of inculcation and savings]\label{dynolaws}

In equilibrium, Proposition \ref{twofeatures} implies that the effective capital stock evolves over time according to:
\begin{equation} \label{fulleq1ge}
(1-\alpha)k_t^{-\alpha} = \frac{(1+n)(1+a)}{(1+\beta)\delta}\left(\left(\frac{k_t}{k_{t-1}}\right)^{1-\alpha}+\frac{\tau_e(1+a)}{1-\alpha}\frac{k_{t+1}}{k_{t-1}^{1-\alpha}}\right)
\end{equation}
while the inculcative parameter $\eta_t$ evolves over time according to:

\begin{equation} \label{fulleq2ge}
\eta_{t+1} = \frac{\beta}{(1+\beta)\delta} - \frac{k_{t+1}}{k_{t}^{1-\alpha}}\left(\frac{\tau_e(1+a)}{(1-\alpha)(1+\beta)\delta}+\frac{(1+n)(1+a)}{\alpha\delta(1-\gamma_e)}\right)
\end{equation}
\end{proposition}

\begin{proof}[Proof: ]
See Appendix \ref{appd}.
\end{proof}

Equation (\ref{fulleq2ge}) describes the dynamics of inculcation but also determines whether there will be any inculcative investment at all in equilibrium. When the right-hand side of (\ref{fulleq2ge}) is nonpositive, $\eta_{t+1}=0$ in equilibrium. This possibility depends on the rates of population growth and technological progress; in particular, a higher rate of technological progress makes $\eta=0$ more likely. Higher growth rates increase the demand and hence returns on capital, making inculcation less desirable. Also making $\eta=0$ more likely in equilibrium are higher inculcative costs, as measured by $\delta$, and greater direct participation in the economy by the elderly (larger values for both $\tau_e$ and $\gamma_e$). 

A value of the relative consumption ratio of the elderly $\frac{c_{et}}{c_{mt}}$ can be found using the equilibrium values of savings, inculcation, and gifts. In Appendix \ref{appd}, we describe the calculations that lead to Proposition \ref{finalrc}:

\begin{proposition}[Relative consumption, inculcation, and growth]\label{finalrc} In equilibrium, when $\eta_t^*>0$, the relative ratio of elderly to middle-age consumption is:
\begin{equation} \label{eqconsratrf}
    \frac{c_{et}}{c_{mt}}=\frac{\beta R_t}{1-\eta_t}\frac{y_{mt-1}+\frac{y_{et}}{R_t}}{y_{mt}+\frac{y_{et+1}}{R_{t+1}}}
\end{equation}
Otherwise, relative consumption is simply determined by relative incomes; that is
\begin{equation}
\frac{c_{et}}{c_{mt}}=\frac{y_{et}+R_ts_{t-1}}{y_{mt}}
\end{equation}
\end{proposition}
\begin{proof}[Proof: ] 
See Appendix D.
\end{proof}

Proposition \ref{finalrc} indicates that the relative consumption ratio depends on relative lifetime incomes, and therefore the overall growth rate of the economy. The proposition also suggests that equilibrium asset returns, $R$, and the magnitude of $\eta$ are important in driving relative consumption. This is particularly true along a balanced growth path, as in a balanced growth equilibrium, incomes grow at a rate $a$, and asset returns are constant. Along a balanced growth path, (\ref{eqconsratrf}) reduces to:
\begin{equation} \label{eqconsratbgss}
       \frac{c_{e}}{c_{m}}=\frac{\beta R}{(1-\eta)(1+\alpha)}
\end{equation}
Equation (\ref{eqconsratbgss}) therefore indicates that parameters such as $\alpha$, the labor intensity of production, $n$, the growth rate of population, and $\delta$, the per-unit-income costs of inculcation will influence the relative consumption of the elderly by exerting an impact on the level of inculcation $\eta$ and equilibrium asset returns $R$. Of course, $\eta$ and $R$ are somewhat complicated functions of model parameters, as discussed in Section \ref{fullsol}. 

\subsection{Analysis with accumulation}

The rest of the modeling section of the paper consists in delving into the implications of Propositions \ref{twofeatures}, \ref{dynolaws}, and \ref{finalrc}. We first solve a simplified version of the model in which all elderly income derives from asset accumulation, and $y_e=0$ at every $t$, which is equivalent to assuming that $\tau_e=\gamma_e=0$, so the elderly do not contribute any productive labor and only provide previously accumulated capital to the economy. 

\subsubsection{Asset income as the sole source of elderly income}

When $y_{et}=0$, the equilibrium gift function (\ref{reactfun}) reduces to:
\begin{equation} \label{simpgift}
    g_{mt}^*=\frac{\eta_t}{1+\beta}y_{mt}
\end{equation}

And the rate equilibration condition (\ref{fulleq2}) is:

\begin{equation} \label{eqincpathsim}
R_{t}=\frac{1+n}{1+\beta}\frac{y_{mt}}{\delta y_{mt-1}}
\end{equation}

The sum of cultural investment and asset savings (\ref{fulleq1}) becomes:
\begin{equation} \label{fulleq2a}
s_{mt}+\eta_{t+1}d_t=\frac{\beta}{1+\beta} y_{mt}
\end{equation}

Equations (\ref{simpgift}, \ref{eqincpathsim} and \ref{fulleq2a}) reflect that when there is no elderly labor income, the decisions of current middle-ageds reduce to allocating shares of income between overall savings, gift-giving and inculcation, in such a way that the returns between inculcative investment and asset accumulation are equilibrated. 

The general equilibrium aspects of the model can be deduced by substituting $\tau_e=0$ into (\ref{fulleq1ge}) in Proposition \ref{dynolaws}, which gives a simple recursion for the stock of capital per effective unit labor:
\begin{equation} \label{sskrel}
    k_{t+1}=\frac{(1-\alpha)(1+\beta)\delta}{(1+n)(1+a)}k_t^{1-\alpha}
\end{equation}

Using $\gamma_e=0$ and $\tau_e=0$ in (\ref{fulleq2ge}) in Proposition \ref{dynolaws} and rearranging results in another recursion for the effective capital stock:

\begin{equation} \label{olgrel}
    k_{t+1}=\frac{\alpha}{(1+a)(1+n)}\left(\frac{\beta}{1+\beta}-\delta\eta_{t+1}\right)k_t^{1-\alpha}
\end{equation}

Equation (\ref{olgrel}) shows how the possibility of inculcation and gifts impacts the canonical overlapping-generations model with population and technological progress, as developed in \citet{acemoglu09}. But for the $\delta\eta_{t+1}$ term in (\ref{olgrel}), the accumulation condition is a standard capital accumulation condition in growth models, so the possibility of inculcative investment reduces the amount of capital per effective unit of labor, which will increase equilibrium asset return rates. Dividing (\ref{sskrel}) by (\ref{olgrel}) and solving the result for $\eta_{t+1}$ gives a time-independent solution for the inculcative parameter $\eta$:

\begin{equation} \label{etasta}
    \eta_{t+1}^*=\eta^*=\frac{\beta}{(1+\beta)\delta}-\frac{1-\alpha}{\alpha}(1+\beta)
\end{equation}

Equation (\ref{etasta}) indicates that a larger $\delta$ reduces the equilibrium size of $\eta^*$. An increase in the relative intensity of capital in production, measured by $\frac{1-\alpha}{\alpha}$, also negatively impacts $\eta$. Societies that employ more capital-intensive means of production, other things equal, will have lower degrees of inculcation.  

If parameters are configured so that $\eta^* \leq 0$, in equilibrium $\eta=0$, and there is no inculcative investment. The elderly ensure their livelihood as they would in the typical overlapping generations model: through prior accumulation of productive assets. Setting (\ref{etasta}) to zero and solving for $\alpha$ gives an interpretation of (\ref{etasta}) as a condition for inculcation to occur in terms of the relative capital intensity of production:

\begin{equation} \label{inculcdec}
\frac{1-\alpha}{\alpha} \leq \frac{\beta}{\delta(1+\beta)^2}
\end{equation}

This threshold depends positively on the costs of inculcation, and on $\beta$, which influences the relative weight placed on consumption in old age in utility. If (\ref{inculcdec}) is not satisfied and no inculcation occurs, we recover the law of motion for the capital stock by using $\eta=0$ in (\ref{olgrel}). In forming the relative ratio of consumption, we find it most straightforward to solve the model in terms of $R$, equilibrium asset returns, rather than in terms of $k$, the equilibrium per unit of labor capital stock. We find a steady-state value of $k$, $\overline{k}$, and then use this to compute the balanced-growth rate of asset returns  $\overline{R}=(1-\alpha)\overline{k}^{-\alpha}$. 

When no inculcation occurs (that is, (\ref{inculcdec}) is not satisfied), equilibrium returns can be found by using $\eta=0$ in (\ref{olgrel}). Writing the resultant expression in terms of $\overline{R}$ gives:

\begin{equation} \label{noincssk}
    \overline{R} = (1-\alpha)\frac{(1+a)(1+n)(1+\beta)}{\beta\alpha}
\end{equation}

When (\ref{inculcdec}) is satisfied, we find $\eta^*$ as given in (\ref{etasta}), and equilibrium asset returns follow from (\ref{olgrel}). These are:
\begin{equation}
    \overline{R} = \frac{(1+a)(1+n)}{(1+\beta)\delta}
\end{equation}

Collecting all of these results gives us a relative consumption ratio that depends on whether or not inculcation and gift giving will occur in equilibrium, which can be obtained by using equilibrium values of $\overline{R}$ and $\eta$ in Equation (\ref{eqconsratbgss}):

\begin{equation} \label{cecmssrat}
\begin{aligned}
\frac{c_e}{c_m} =&\frac{1-\alpha}{\alpha}(1+\beta)(1+n), & \eta^*=0\\ & \frac{\beta(1+n)}{\delta(1+\beta)(1+\frac{1-\alpha}{\alpha}(1+\beta))-\beta}, 
                & \eta^*>0
\end{aligned}
\end{equation}

Population growth increases the relative well-being of the elderly, whether or not gifts are involved. In a high-capital-intensity setting without gift-giving or inculcation, this occurs because population growth raises the relative scarcity of capital, boosting returns on accumulated assets and benefiting the elderly. In cases with gift-giving, this effect is amplified by the greater number of young people, leading to more gifts being distributed to the elderly.  

The capital intensity of production, \(\frac{1-\alpha}{\alpha}\), significantly influences the relative treatment of the elderly. As in the model from Section \ref{ge}, more labor-intensive production tends to lower the relative treatment of the elderly. However, cultural practices such as inculcation and acculturation emerge to mitigate this decline when income primarily depends on labor rather than accumulated assets, providing a way to transfer income to the elderly. Conversely, in more asset-intensive production, the elderly's relative treatment improves due to asset income, but cultural acculturation to support the elderly diminishes.  

Figure \ref{capintense} illustrates this dynamic, showing the relative consumption of the elderly, \(\frac{c_e}{c_m}\), from (\ref{cecmssrat}) as a function of capital intensity. The figure reveals a nonlinear, U-shaped relationship between elderly treatment and capital intensity, similar to the pattern seen in the property-rights model of Section \ref{ge}.

\begin{figure}[htbp] 
\begin{center}
\includegraphics[width=.75\textwidth]{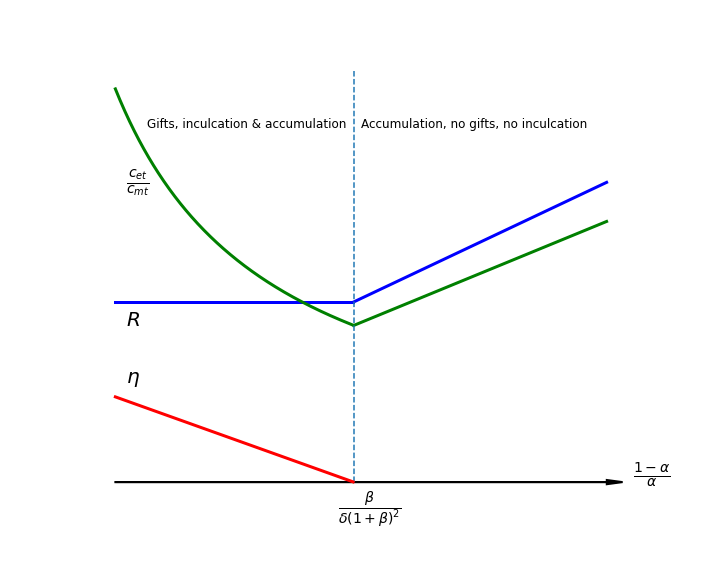}
\caption{Relative consumption of the elderly ($\frac{c_{et}}{c_{mt}}$), as a function of the capital intensity of production, $\frac{1-\alpha}{\alpha}$. } \label{capintense}
\end{center} 
\end{figure}

\subsubsection{Full model solution} \label{fullsol}

When the elderly can save and also have a sources of labor income, the model is more complex because elderly labor income provides sustenance and also raises lifetime income. This makes it difficult to accurately state how changes in model parameters influence inculcation, accumulation, and relative treatment of the elderly. With $\tau_e,\gamma_e>0$, the steady-state relationships corresponding to (\ref{fulleq1ge}) and (\ref{fulleq2ge}) in Proposition \ref{dynolaws} can be written as:
\begin{equation} \label{Rwithfu}
R = \frac{(1+n)(1+a)}{(1+\beta)\delta}\left(1 +\frac{(1+a)\tau_{e}}{ R}\right)
\end{equation}
with
\begin{equation} \label{etawithfu}
\eta = \frac{\beta}{\delta(1+\beta)}-\frac{(1+a)}{\delta R}\left(\frac{\tau_e}{1+\beta}+\frac{(1+n)}{(1-\gamma_e)}\frac{1-\alpha}{\alpha}\right)
\end{equation}

Solving (\ref{Rwithfu}) for the equilibrium asset return rate gives:

\begin{equation}
R = \frac{(1+n)(1+a)}{2\delta(1+\beta)}\left(1+\sqrt{1+\frac{4(1+\beta)\delta}{(1+n)}r_e}\right)
\end{equation}

From (\ref{Rwithfu}), we see that as $\tau_e$ increases, meaning that the elderly have a relatively larger source of income, $R$ increases in an equilibrium with a positive amount of inculcation ($\eta>0$). When a future source of income is available, the incentives to save are deadened. Similarly, increases in $\tau_e$ lessen inculcative investments, as is evident from (\ref{etawithfu}). 

When there is no inculcation, we can solve (\ref{fulleq2ge}) with $\eta=0$ and rearrange terms to get the following expression for asset returns:
\begin{equation} \label{eqRwithnoi}
    R=\frac{1+a}{\beta}\left(\frac{(1+n)(1+\beta)}{1-\gamma_e}\frac{1-\alpha}{\alpha}+\tau_e\right)
\end{equation}

Using (\ref{eqRwithnoi}) and $\eta=0$ in (\ref{eqconsratbgss}), we find that relative consumption becomes:

\begin{equation}
    \frac{c_{e}}{c_{m}}=\frac{(1+n)(1+\beta)}{1-\gamma_e}\frac{1-\alpha}{\alpha}+\tau_e
\end{equation}

We can rearrange expression (\ref{eqconsratbgss}) to give a condition for inculcation to occur once again. This condition again reflects the key role played by the capital intensity of production. No inculcation and gift-giving occur in equilibrium so long as:
\begin{equation}
    \frac{1-\alpha}{\alpha} \geq \left(\frac{R\beta}{(1+a)}-\tau_e\right)\frac{1-\gamma_e}{(1+\beta)(1+n)}
\end{equation}

\begin{figure}[htbp] 
\begin{center}
\includegraphics[width=\textwidth]{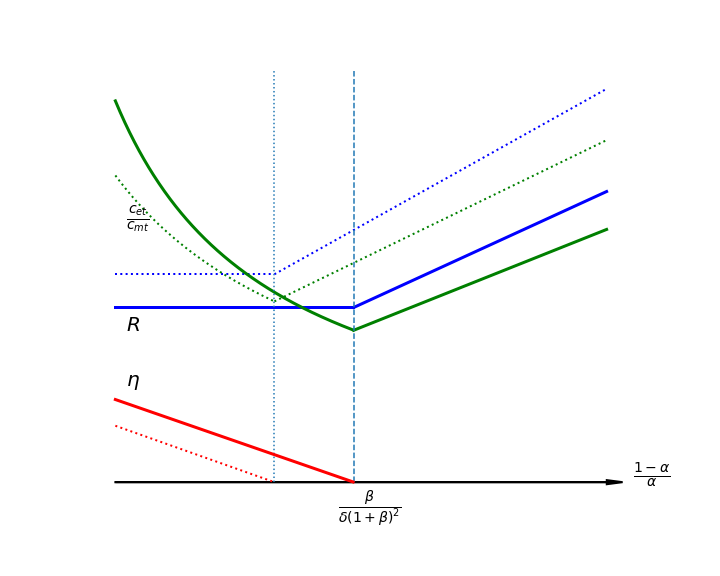}
\caption{Relative consumption of the elderly ($\frac{c_{et}}{c_{mt}}$), as a function of relative capital intensity with $\tau_e>0$. The solid lines show how the rate of return on capital is constant while $\eta$ is positive. When capital intensity reaches a certain threshold, inculcation ends and the rate of returns starts rising. The dashed lines show how these relationships change when the elderly provide productive labor $\tau_e>0$. } \label{taucompstat}
\end{center} 
\end{figure}

It is easiest to see how equilibrium is influenced when $\tau_e>0$ compared to the case in which $\tau_e=0$. as depicted on Figure \ref{capintense}. In Figure \ref{taucompstat}, we superimpose the critical line depicting the relative capital intensity at which $\eta$ goes to zero, and the resulting consumption ratios when $\tau_e>0$ on those depicted in \ref{capintense}. 

Although Figure \ref{taucompstat} suggests greater elderly labor income reduces the incidence of inculcative, gift-giving equilibria, it also indicates that the impact of labor income on  relative elderly consumption are not obvious. On the left-hand side of the figure we see that the relative consumption ratio of the elderly goes down as a result of the increased income shift. This is because total labor supply increases, which depresses labor incomes in general. 

This last set of results indicates how the advent or expansion of societal welfare targeted towards the elderly may impact the transmission and maintenance of culture, some of the particulars of which we described in Section \ref{ethnosec}. Such changes act to increase $\tau_e$ in the model, which may have a deleterious effect on inculcation, and therefore cultural transmission. 

\section{Discussion and Conclusions} \label{conclusion}

The model in this paper builds on prior research into intergenerational transfers and their role in supporting the elderly while allowing for a deeper exploration of cross-cultural findings on elderly treatment. Our analysis spans both recent phenomena, such as the demographic transition, and historical shifts in elderly treatment across varying levels of technological sophistication. By incorporating the interplay between property rights, culture, and elderly treatment, the model centers attention on elderly well-being rather than intermediate economic features like transfer patterns.

The model highlights the complexity of relationships between ownership, accumulation, production, and elderly treatment, providing theoretical support for observed ethnographic patterns including the `curvilinear hypothesis.' It also addresses how accumulation in production impacts the elderly, showing that while initial accumulation may disadvantage them, increasing its importance can eventually benefit their relative status. These findings align with the hypothesis of a nonlinear relationship between familial complexity and degree of modernization \citep{nimkoff60, blumberg72}. They echo \citet{murdock73}, who observed that the most culturally advanced societies often fall at the mid-spectrum of technological complexity, noting, ``This essentially bimodal or curvilinear distribution is inconsistent with any unilinear interpretation of social development'' \citep[p. 392]{murdock73}. Our model offers insights into why such a pattern might arise, particularly regarding elderly treatment and cultural transmission.\footnote{Similar patterns are discussed in the context of land inheritance by \citet{bakermiceli2005} and marriage institutions by \citet{bakerjacobsen2007b, bakerjacobsen2007a}.}

Several avenues for future research emerge from this work. One is the interplay between family structure, technological change, and elderly treatment. While our model touches on these issues through its focus on gift-giving and familial ties, deeper exploration of the connections between elderly treatment, marriage, and fertility is warranted.\footnote{See related discussions in \citet{cigno20} and \citet{bisver00}.} Gender differences in elderly treatment also deserve attention, as emphasized by \citet{simmons70}. Furthermore, the dynamics of inculcation and cultural transmission, including their public-goods and behavioral aspects, remain ripe for investigation \citep{bowles98}. Finally, while our focus has been on inter-societal well-being of the elderly, future research could address inequality among the elderly, examining how poverty and disparities relate to cultural transmission and gift-giving practices. We hope this framework will inspire further inquiry into these important topics.

\newpage
\begin{appendices}
\section{Data Appendix} \label{datapp}
\setcounter{table}{0}
\renewcommand{\thetable}{A\arabic{table}}
We describe the data underlying the analysis presented in Tables 1 and 2 of the text and also provide an overview of the data in \citet{simmons70}.\footnote{All data are derived directly from \citet{simmons70}, and are available on request from the authors.} Simmons developed variables to measure the presence and importance of material, social, and cultural characteristics across 71 cultures, including hunters, pastoralists, nomads, and agriculturalists. Features were coded as $+$ (very important), $=$ (somewhat important), $-$ (present but unimportant), or $0$ (absent). We recoded nonmissing values numerically, ranging from 3 (very important) to 0 (absent).

Simmons's data includes 224 features, with some values specific to women, men, or the elderly. However, missing values and vague variable definitions make formal analysis beyond Simmons's cross-tabulations challenging and potentially misleading. Therefore, we adopted a descriptive approach, compiling summary indices from his data. The evidence presented in the text is based on these indices, whose construction we now discuss.

Table \ref{openeco}1 describes the variables that we judge as capturing the market-openness of the economic environment present; specifically, the degree to which property rights, exchange, and market orientation are features of a society.   

\begin{table}[H] 
\centering
\begin{small}
\include{openeco}\label{openeco}\caption{Variables included in the openness of economy index, which captures the degree to which a market-oriented economy might be present in each society}
\end{small}
\end{table} 

Our market openness index is constructed by simply adding and subtracting characteristics. While more complex methods like principal components could be used, experimentation showed little benefit compared to this transparent approach. We sum the first six characteristics in Table \ref{openeco}1 and subtract the last (communal land ownership). Summary statistics for the index are in Table \ref{indices}8.

Simmons also notes how actively the elderly participate in aspects of the economy. The variables summarized in Table \ref{dirprod}2 include jobs that the elderly do in society that we judge to directly support the production of physical output in a society, or perform functions that free other members of society to engage in direct production. 

\begin{table}[H] 
\centering
\begin{small}
\include{directprodvars}\label{dirprod}\caption{Variables included in market development and sophistication index}
\end{small}
\end{table} 

From these variables, we form what we call a direct production index by simply adding all the variables in the table, treating missing values as zeros. The resulting index is summarized in Table \ref{indices}8 and is designed to capture something about the degree to which the elderly contribute to material life in each culture. 

The next set of variables we use to construct an index of the degree to which the elderly perform important ritualistic or social functions. These variables are summarized on table \ref{socfuns}3; the sum of the variables comprises our index of social functions, also summarized in table \ref{indices}8.

\begin{table}[H]
\centering
\begin{small}
\include{socprodvars} \label{socfuns}\caption{Social functions performed by the elderly}
\end{small}
\end{table}

The variables summarized in Table \ref{knowprov} 4 describe the degree to which the elderly provide advice, store important social information, or perform other functions that require knowledge. We computed an index of knowledge provision based on just summing up these variables.

\begin{table}[H]
\centering
\begin{small}
\include{knowprodvars}\label{knowprov}\caption{Knowledge provision by the elderly}
\end{small}
\end{table}

The variables in Table \ref{titular}5 describe the degree to which the elderly hold important titles in society, function as elders, judges, or chiefs. Our index of titular duties is created as a simple sum of these variables. 

\begin{table}[H]
\centering
\begin{small}
\include{titulars}\label{titular}\caption{Titular functions of the elderly}
\end{small}
\end{table}

Various variables catalog the degree to which the elderly are respected, supported by the community, or more generally viewed by society. These variables are collected in Table \ref{treatculc}6, and we collect these variables into an index that is a simple sum of these variables. 

\begin{table}[H]
\centering
\begin{small}
\include{treatment}\label{treatculc}\caption{Positive treatment practices towards the elderly}
\end{small}
\end{table}

Some variables actually measure the degree to which the elderly are pictured as having magical powers and related beliefs. We gather all of these variables into an index, which is a simple sum of the variables in Table \ref{treatculc}7. 

\begin{table}[H]
\centering
\begin{small}
\include{inculcation}\label{inculcation}\caption{Inculcated beliefs about the elderly}
\end{small}
\end{table}

Table \ref{indices}8 contains a summary of all indices that we have created. These indices are used to calculate the correlations in Table 1 in the text.

\begin{table}[H]
\centering
\begin{small}
\include{indices} \caption{Summary statistics for constructed indices used in Table 1.}
\end{small}
\end{table} \label{indices}

\section{Proof of Proposition \ref{giftfunny}} \label{appa}
\renewcommand{\theequation}{B.\arabic{equation}}
\setcounter{equation}{0}

The solution to (\ref{maxprob}) starts by assuming an interior solution and then checking when this solution works. Assume that agents maintain the belief that gifts are a multiplicative function of the value of $\eta_t$:
\begin{equation} \label{Psi}
    g_{mt+1}=\eta_{t+1}\Psi_{t+1}
\end{equation}

where $\Psi_{t+1}$ is to be determined. Substituting for $g_{mt+1}$ using (\ref{Psi}) into the middle-aged agent's optimization problem (\ref{maxprob}) and differentiating gives first-order conditions:

\begin{equation} \label{focs}
\begin{aligned}
  &  -\frac{1-\eta_t}{y_{mt}-d_t\eta_{t+1}-s_{mt}-g_{mt}}+\frac{\eta_t}{g_{mt}}&=0 \\
  &  -\frac{1-\eta_t}{y_{mt}-d_t\eta_{t+1}-s_{mt}-g_{mt}}+\beta\frac{R_{t+1}}{s_{mt}R_{t+1}+(1+n)g_{et+1}+y_{et+1}}&=0 \\
& -\frac{(1-\eta_t)d_t}{y_{mt}-d_t\eta_{t+1}-s_{mt}-g_{mt}}+\beta\frac{(1+n)\Psi_{t+1}}{s_{mt}R_{t+1}+(1+n)g_{et+1}+y_{et+1}}&=0 \\
\end{aligned}
\end{equation}

Rearranged, the first expression in (\ref{focs}) implies that gifts are proportional to income:

\begin{equation} \label{grf}
    g_{mt} = \eta_t(y_{mt}-d_t\eta_{t+1}-s_{mt})
\end{equation}

Substituting $g_{mt}$ from (\ref{grf}) back into (\ref{focs}) gives two conditions describing savings and inculcation behavior:

\begin{equation} \label{focs2}
\begin{aligned}
  &  -\frac{1}{y_{mt}-d_t\eta_{t+1}-s_{mt}}+\beta\frac{R_{t+1}}{s_{mt}R_{t+1}+(1+n)g_{et+1}+y_{et+1}}&=0 \\
& -\frac{d_t}{y_{mt}-d_t\eta_{t+1}-s_{mt}}+\beta\frac{(1+n)\Psi_{t+1}}{s_{mt}R_{t+1}+(1+n)g_{et+1}+y_{et+1}}&=0 \\
\end{aligned}
\end{equation}

Equations (\ref{focs2}) imply that for both savings and inculcation to occur, it must be that:

\begin{equation} \label{eqrel}
    \Psi_{t+1}=\frac{d_tR_{t+1}}{1+n}
\end{equation}

Multiplying both sides of (\ref{eqrel}) by $\eta_{t+1}$ gives $g_{mt+1}$ in \ref{Psi}:

\begin{equation} \label{eqgift}
    g_{mt+1}=\eta_{t+1}\frac{d_tR_{t+1}}{(1+n)}
\end{equation}

Since $g_{et+1}=(1+n)g_{mt+1}$, we find that $g_{et+1}=\eta_{t+1}d_tR_{t+1}$. Using this in either part of (\ref{focs2}) gives an expression describing how total savings are distributed between asset accumulation $s_{mt}$ and inculcation $\eta_{t+1}$, which is  equation (\ref{fulleq1}) in the text:

\begin{equation} \label{totsave}
s_{mt}+\eta_{t+1}d_t=\frac{\beta y_{mt}-\frac{y_{et+1}}{R_{t+1}}}{1+\beta}
\end{equation}

Inserting the right-hand side of (\ref{totsave}) into the right-hand side of (\ref{grf}) gives the equilibrium gift function in the text, equation (\ref{reactfun}):

\begin{equation} \label{gstarc}
    g^*_{mt}=\frac{\eta_t}{1+\beta}\left(y_{mt}+\frac{y_{et+1}}{R_{t+1}}\right)
\end{equation}

Comparing (\ref{Psi}) and (\ref{gstarc}), we see that
\begin{equation} \label{eqpsi}
    \Psi_t = \frac{y_{mt}+\frac{y_{et+1}}{R_{t+1}}}{1+\beta}
\end{equation}

\section{Proof of Proposition \ref{twofeatures}} \label{appb}
\renewcommand{\theequation}{C.\arabic{equation}}
\setcounter{equation}{0}

The first part of the proposition concerning aggregate savings follows directly from Appendix \ref{appa}, equation (\ref{totsave}). To get the second part of the Proposition, iterate (\ref{eqrel}) back one period and substitute the value of $\Psi_t$ into (\ref{eqpsi}) to get:
\begin{equation} \label{eqrelfin}
    R_t = \left(\frac{1+n}{1+\beta}\right)\frac{y_{mt}+\frac{y_{et+1}}{R_{t+1}}}{d_{t-1}}
\end{equation}
Which (\ref{eqrelfin}) appears in the text as (\ref{fulleq2}).

\section{Proof of Proposition \ref{dynolaws}} \label{appc}
\renewcommand{\theequation}{D.\arabic{equation}}
\setcounter{equation}{0}

We begin by rewriting (\ref{fulleq1}) as follows, using the assumption that $d_t=\delta y_{mt}$:

\begin{equation} \label{bs1}
    \eta_{t+1} = \frac{\beta}{(1+\beta)\delta}-\frac{y_{et+1}}{\delta y_{mt}}\frac{1}{R_{t+1}(1+\beta)}-\frac{s_{mt}}{\delta y_{mt}}
\end{equation}

Substitution for $R_{t+1}$, $y_{mt}$ and $y_{et+1}$ using (\ref{wagesrates}) in (\ref{bs1}), and also noting that the future capital stock depends on savings by middle-aged agents in the present, $K_{t+1}=s_{mt}N_{mt}$, gives:

\begin{equation} \label{bs2}
\eta_{t+1} = \frac{\beta}{(1+\beta)\delta}-\frac{ A_{et+1}k_{t+1}}{(1-\alpha)(1+\beta)\delta A_{mt}k_t^{1-\alpha}}-\frac{K_{t+1}}{N_{mt}\delta \alpha A_{mt}k_t^{1-\alpha}}
\end{equation}

Multiplying the last term on the right by $L_{t+1}/L_{t+1}$, and substituting $N_{mt+1}=(1+n)N_{mt}$, $A_{mt+1}=(1+a)A_{mt}$, and $A_{et+1}=(1+a)A_{et}$ into (\ref{bs2}), we have:

\begin{equation} \label{bs3}
\eta_{t+1} = \frac{\beta}{(1+\beta)\delta}-\frac{ A_{et}(1+a)k_{t+1}}{(1-\alpha)(1+\beta)\delta A_{mt}k_t^{1-\alpha}}-\frac{K_{t+1}}{L_{t+1}}\frac{(1+a)(1+n)}{\frac{N_{mt+1}A_{mt+1}}{L_{t+1}}\delta \alpha k_t^{1-\alpha}}
\end{equation}

Collecting terms and using $\tau_e=\frac{A_e}{A_m}$, and noting that $\frac{N_{mt+1}A_{mt+1}}{L_{t+1}}$, the middle-aged share in total labor supply is constant over time and equal to $1-\gamma_e$, one less the elderly share, we have:

\begin{equation} \label{bs4}
\eta_{t+1} = \frac{\beta}{(1+\beta)\delta} - \frac{k_{t+1}}{k_{t}^{1-\alpha}}\left(\frac{\tau_e(1+a)}{(1-\alpha)(1+\beta)\delta}+\frac{(1+n)(1+a)}{\alpha\delta(1-\gamma_e)}\right)
\end{equation}

This appears as (\ref{fulleq2ge}) in the text. For equilibrium relationship (\ref{fulleq1ge}), we begin by substituting the expressions in (\ref{wagesrates}) into (\ref{fulleq2}) to get: 
\begin{equation} \label{bs5}
R_t = \frac{1+n}{1+\beta}\left(\frac{ A_{mt}k_t^{1-\alpha}}{\delta  A_{mt-1}k_{t-1}^{1-\alpha}}+\frac{\frac{ A_{et+1}k_{t+1}^{1-\alpha}}{R_{t+1}}}{\delta A_{mt-1}k_{t-1}^{1-\alpha}}\right)
\end{equation}
Using $A_{mt+1}=(1+a)A_{mt}$, $A_{mt-1}=\frac{A_{mt+1}}{(1+a)^2}$, and $\tau_e=\frac{A_{et}}{A_mt}$, gives:
\begin{equation} \label{bs6}
    R_t = \frac{1+n}{1+\beta}\left(\frac{1+a}{\delta}\left(\frac{k_t}{k_{t-1}}\right)^{1-\alpha}+\frac{r_e(1+a)^2}{R_{t+1}\delta}\left(\frac{k_{t+1}}{k_{t-1}}\right)^{1-\alpha}\right)
\end{equation}

Putting in (\ref{bs6}) the expression for asset returns, $R_t=(1-\alpha)k_t^{-\alpha}$, gives expression (\ref{fulleq1ge}) in the text. One convenience of equation (\ref{bs6}) is that it yields a simple steady-state relationship for asset returns:

\begin{equation} \label{bs7}
    R = \frac{1+n}{1+\beta}\left(\frac{1+a}{\delta}+\frac{(1+a)^2r_{e}}{\delta R}\right)
\end{equation}

which easily factors into equation (\ref{Rwithfu}) in the text. 

\renewcommand{\theequation}{E.\arabic{equation}}
\setcounter{equation}{0}
\section{Proof of Proposition \ref{finalrc}} \label{appd}

We derive the relative consumption ratio for both the case in which inculcation occurs and when it does not, starting with the latter case. In this case, elderly consumption is current income plus income earned from asset accumulation in the prior period:
\begin{equation} \label{c1}
    c_{et}=y_{et}+s_{mt-1}R_t
\end{equation}
While for middle-aged agents, if there is no inculcative investment, current consumption is current income less savings:
\begin{equation} \label{c2}
    c_{mt} = y_{mt}-s_{mt}
\end{equation}
Moreover, one can use (\ref{fulleq1}) to conclude that savings is determined by: 
\begin{equation} \label{savratac}
    s_{mt}=\frac{\beta y_{mt} - \frac{y_{et+1}}{R_{t+1}}}{1+\beta}
\end{equation}
Plugging (\ref{savratac}) into (\ref{c2}) gives:
\begin{equation} \label{c3}
    c_{mt}=\frac{y_{mt} + \frac{y_{et+1}}{R_{t+1}}}{1+\beta}
\end{equation}
One can also use (\ref{savratac}) to deduce that elderly consumptio in (\ref{c1}) can be written as:
\begin{equation} \label{c4}
    c_{et}=\frac{\beta}{1+\beta}\left(y_{et}+R_ty_{mt-1}\right)
\end{equation}
Forming the ratio of expressions and factoring out $R_t$ gives:
\begin{equation} \label{autconsrat}
    \frac{c_{et}}{c_{mt}}=\beta R_t \left(\frac{y_{mt-1}+\frac{y_{et}}{R_t}}{y_{mt}+\frac{y_{et+1}}{R_{t+1}}}\right)
\end{equation}

In the case where there is inculcation, we observe that consumption for the elderly is now:
\begin{equation} \label{c5}
c_{et} = y_{et} + (1+n)g_{mt} + s_{mt-1}R_t 
\end{equation}
while consumption for the middle-aged is:
\begin{equation} \label{c6}
c_{mt} = y_{mt} -g_{mt}- s_{mt} - d_t\eta_{t+1} 
\end{equation}
Elderly consumption is most easily found using (\ref{focs2}), which allows us to conclude that:
\begin{equation} \label{c7}
    s_{mt-1}R_t + (1+n)g_{mt} + y_{et} = \beta R_{t}(y_{mt-1}-d_{t-1}\eta_t-s_{mt-1})
\end{equation}
Using (\ref{c7}) in (\ref{c5}) gives:
\begin{equation} \label{c71}
    c_{et}=R_{t}(y_{mt-1}-d_{t-1}\eta_t-s_{mt-1})
\end{equation}
Eliminating gifts in (\ref{c6}) using (\ref{reactfun}) gives:
\begin{equation} \label{c8}
c_{mt} = (1-\eta)(y_{mt} - s_{mt} - d_t\eta_{t+1} )
\end{equation}
For both (\ref{c7}) or (\ref{c8}), we can apply (\ref{fulleq1}). Then, forming the relative consumption ratio gives:
\begin{equation} \label{c9}
\frac{c_{et}}{c_{mt}}=\beta R_t\frac{y_{mt-1}+\frac{y_{et}}{R_{t}}}{y_{mt}+\frac{y_{et+1}}{R_{t+1}}}
\end{equation}
which is expression (\ref{eqconsratrf}) in the text.

\end{appendices}

\newpage
\clearpage
\singlespacing
\bibliographystyle{itaxpf} 
\bibliography{sources} 

\clearpage
\end{document}

%% file: openeco.tex
{
\def\sym#1{\ifmmode^{#1}\else\(^{#1}\)\fi}
\begin{tabular}{l*{1}{cccccc}}
\toprule
                    &            &        Mean&          SD&         Min&         Max&           N\\
\midrule
Trade presence      &            &        2.33&        0.75&           1&           3&          67\\
Money in use        &            &        1.70&        1.07&           0&           3&          57\\
Private ownership of objects&            &        2.60&        0.71&           1&           3&          63\\
Private ownership of land&            &        1.37&        1.19&           0&           3&          54\\
Male property rights&            &        2.58&        0.72&           1&           3&          53\\
Female property rights&            &        1.56&        0.98&           0&           3&          32\\
Communal ownership of land&            &        2.21&        0.76&           0&           3&          56\\
\bottomrule
\end{tabular}
}

%% file: directprodvars.tex
{
\def\sym#1{\ifmmode^{#1}\else\(^{#1}\)\fi}
\begin{tabular}{l*{1}{cccccc}}
\toprule
                    &            &        Mean&          SD&         Min&         Max&           N\\
\midrule
Male traders        &            &        2.35&        0.75&           1&           3&          20\\
Female traders      &            &        1.33&        0.52&           1&           2&           6\\
Male provision of infant services&            &        1.79&        0.93&           0&           3&          24\\
Female provision of infant services&            &        2.49&        0.83&           0&           3&          47\\
Male education of young&            &        2.83&        0.45&           1&           3&          40\\
Female education of young&            &        2.86&        0.35&           2&           3&          36\\
Male nurses         &            &        2.90&        0.31&           2&           3&          20\\
Female nurses       &            &        2.67&        0.72&           1&           3&          15\\
Males aid with agriculture&            &        1.82&        1.38&           0&           3&          17\\
Females aid with agriculture&            &        2.15&        1.31&           0&           3&          20\\
Males aid with hunting&            &        2.53&        0.62&           1&           3&          17\\
Females aid with hunting&            &        2.67&        0.52&           2&           3&           6\\
Males aid with herding&            &        1.29&        1.44&           0&           3&          14\\
Females aid with herding&            &        0.83&        1.27&           0&           3&          12\\
Males aid with collecting&            &        2.50&        0.71&           2&           3&           2\\
Females aid with collecting&            &        2.67&        0.52&           2&           3&           6\\
Males aid in household&            &        2.55&        0.52&           2&           3&          11\\
Females aid in household&            &        2.84&        0.37&           2&           3&          32\\
Males make toys and tools&            &        2.36&        0.63&           1&           3&          14\\
Females make toys and tools&            &        2.17&        1.17&           0&           3&           6\\
Male elderly fees for services&            &        2.80&        0.41&           2&           3&          25\\
Female elderly fees for services&            &        2.56&        0.62&           1&           3&          18\\
\bottomrule
\end{tabular}
}

%% file: socprodvars.tex
{
\def\sym#1{\ifmmode^{#1}\else\(^{#1}\)\fi}
\begin{tabular}{l*{1}{cccccc}}
\toprule
                    &            &        Mean&          SD&         Min&         Max&           N\\
\midrule
Male rites of initiation&            &        2.40&        1.00&           0&           3&          25\\
Female rites of initiation&            &        1.90&        0.94&           0&           3&          21\\
Males perform marriage rites&            &        2.77&        0.43&           2&           3&          31\\
Females perform marriage rites&            &        2.43&        0.60&           1&           3&          21\\
Males perform funereal rites&            &        2.60&        0.76&           0&           3&          25\\
Females perform funereal rites&            &        2.54&        0.51&           2&           3&          26\\
Male priesthood     &            &        2.86&        0.36&           2&           3&          35\\
Female priesthood   &            &        1.85&        0.69&           0&           3&          13\\
Male shamans        &            &        2.86&        0.39&           1&           3&          59\\
Female shamans      &            &        2.49&        0.69&           0&           3&          47\\
Male decorators     &            &        2.33&        0.58&           2&           3&           3\\
Female decorators   &            &        2.67&        0.50&           2&           3&           9\\
\bottomrule
\end{tabular}
}

%% file: knowprodvars.tex
{
\def\sym#1{\ifmmode^{#1}\else\(^{#1}\)\fi}
\begin{tabular}{l*{1}{cccccc}}
\toprule
                    &            &        Mean&          SD&         Min&         Max&           N\\
\midrule
Males provide expert advice&            &        2.84&        0.37&           2&           3&          25\\
Females provide expert advice&            &        2.71&        0.49&           2&           3&           7\\
Male stores of social information&            &        2.98&        0.14&           2&           3&          54\\
Female stores of social information&            &        2.59&        0.50&           2&           3&          27\\
Males lead festivals&            &        2.69&        0.52&           1&           3&          42\\
Females lead festivals&            &        2.50&        0.64&           1&           3&          28\\
\bottomrule
\end{tabular}
}

%% file: titulars.tex
{
\def\sym#1{\ifmmode^{#1}\else\(^{#1}\)\fi}
\begin{tabular}{l*{1}{cccccc}}
\toprule
                    &            &        Mean&          SD&         Min&         Max&           N\\
\midrule
Male chiefs         &            &        2.37&        0.71&           1&           3&          62\\
Female chiefs       &            &        0.35&        0.61&           0&           2&          31\\
Male council of elders&            &        2.60&        0.65&           0&           3&          57\\
Male judges         &            &        2.56&        0.69&           0&           3&          45\\
Female judges       &            &        0.62&        0.87&           0&           2&          13\\
\bottomrule
\end{tabular}
}

%% file: treatment.tex
{
\def\sym#1{\ifmmode^{#1}\else\(^{#1}\)\fi}
\begin{tabular}{l*{1}{cccccc}}
\toprule
                    &            &        Mean&          SD&         Min&         Max&           N\\
\midrule
Male elderly respected&            &        2.81&        0.47&           1&           3&          63\\
Female elderly respected&            &        2.29&        0.81&           0&           3&          51\\
Male elderly community support&            &        2.09&        0.84&           0&           3&          33\\
Female elderly community support&            &        1.97&        0.81&           0&           3&          30\\
Male elderly family support&            &        2.80&        0.48&           1&           3&          60\\
Female elderly family support&            &        2.58&        0.57&           1&           3&          57\\
Male son-in-law elderly support&            &        1.82&        1.05&           0&           3&          22\\
Female son-in-law elderly support&            &        2.06&        1.06&           0&           3&          16\\
Male elderly seen as distinguished&            &        2.71&        0.61&           1&           3&          14\\
Female elderly seen as distinguished&            &        3.00&           .&           3&           3&           1\\
\bottomrule
\end{tabular}
}

%% file: inculcation.tex
{
\def\sym#1{\ifmmode^{#1}\else\(^{#1}\)\fi}
\begin{tabular}{l*{1}{cccccc}}
\toprule
                    &            &        Mean&          SD&         Min&         Max&           N\\
\midrule
Males glorified as daimons&            &        2.64&        0.54&           1&           3&          36\\
Females glorified as daimons&            &        2.44&        0.65&           1&           3&          25\\
Males glorified as heros&            &        2.77&        0.43&           2&           3&          39\\
Females glorified as heros&            &        2.55&        0.68&           1&           3&          31\\
Males friends of children&            &        2.62&        0.52&           2&           3&           8\\
Females friends of children&            &        2.45&        0.82&           1&           3&          11\\
\bottomrule
\end{tabular}
}

%% file: indices.tex
{
\def\sym#1{\ifmmode^{#1}\else\(^{#1}\)\fi}
\begin{tabular}{l*{1}{cccccc}}
\toprule
                    &            &        Mean&          SD&         Min&         Max&           N\\
\midrule
Openness of economy index&            &        6.49&        4.27&          -3&          14&          71\\
Elderly contribution to direct production index&            &        6.99&        4.35&           0&          17&          71\\
Contribution to social production index&            &        8.72&        3.71&           0&          18&          71\\
Knowledge and advice provision index&            &        3.55&        2.14&           0&           8&          71\\
Titular duties of elderly index&            &        3.02&        1.44&           0&           6&          71\\
Positive treatment of elderly index&            &        5.99&        2.68&           0&          13&          71\\
Positive attitudes to elderly inculcation index&            &        2.75&        2.22&           0&           9&          71\\
Elderly social functions index&            &        5.70&        2.98&           0&          14&          71\\
Total elderly contributions in production index&            &       19.25&        8.88&           2&          40&          71\\
\bottomrule
\end{tabular}
}

%% file: main.bbl
\begin{thebibliography}{62}
\newcommand{\bibenquote}[1]{``#1''}
\providecommand{\natexlab}[1]{#1}

\bibitem[{Acemoglu(2009)}]{acemoglu09}
Acemoglu, D. (2009). \textit{An introduction to modern economic growth}.
  Princeton and Oxford: Princeton University Press.

\bibitem[{Andreoni(1989)}]{andreoni89}
Andreoni, J. (1989). \bibenquote{Giving with impure altruism: Applications to
  charity and ricardian equivalence.} \textit{Journal of Political Economy},
  \textit{97}(6), 1447--1458.

\bibitem[{Arnhoff et~al.(1964)Arnhoff, Leon, and Lorge}]{arnhoff64}
Arnhoff, F.~N., Leon, H.~V., and Lorge, I. (1964). \bibenquote{Cross-cultural
  acceptance of stereotypes towards aging.} \textit{Journal of Social
  Psychology}, \textit{63}(1), 41--58.

\bibitem[{Bagwell(2007)}]{bagwell07}
Bagwell, K. (2007). \bibenquote{Chapter 28: The economic analysis of
  advertising.} In M.~Armstrong, and R.~Porter (Eds.), \textit{Handbook of
  Industrial Organization}, vol.~3, 1701--1844, Elsevier.

\bibitem[{Baker(2003)}]{baker2003}
Baker, M.~J. (2003). \bibenquote{An equilibrium conflict model of land tenure
  in hunter-gatherer societies.} \textit{Journal of Political Economy},
  \textit{111}(1), 124--173.

\bibitem[{Baker(2008)}]{baker2008}
Baker, M.~J. (2008). \bibenquote{A structural model of the transition to
  agriculture.} \textit{Journal of Economic Growth}, \textit{13}(4), 257--292.

\bibitem[{Baker and Conning(2024)}]{bakerconning24}
Baker, M.~J., and Conning, J. (2024). \bibenquote{A model of enclosures:
  Coordination, conflict, and efficiency in the transformation of land property
  rights.} ArXiv: 2311.01592.

\bibitem[{Baker and Jacobsen(2007{\natexlab{a}})}]{bakerjacobsen2007b}
Baker, M.~J., and Jacobsen, J.~P. (2007{\natexlab{a}}). \bibenquote{A
  human-capital based theory of post-marital residence rules.} \textit{Journal
  of Law, Economics, and Organization}, \textit{23}(1), 763--793.

\bibitem[{Baker and Jacobsen(2007{\natexlab{b}})}]{bakerjacobsen2007a}
Baker, M.~J., and Jacobsen, J.~P. (2007{\natexlab{b}}). \bibenquote{Marriage,
  specialization, and the gender division of labor.} \textit{Journal of Labor
  Economics}, \textit{25}, 763--793.

\bibitem[{Baker and Miceli(2005)}]{bakermiceli2005}
Baker, M.~J., and Miceli, T. (2005). \bibenquote{Land inheritance rules: theory
  and cross-cultural analysis.} \textit{Journal of Economic Behavior and
  Organization}, \textit{56}(1), 77--102.

\bibitem[{Baker and Swope(2021)}]{bakerswope21}
Baker, M.~J., and Swope, K.~J. (2021). \bibenquote{Sharing, gift-giving, and
  optimal resource use in hunter-gatherer society.} \textit{Economics of
  Governance}, \textit{22}(2), 119--138.

\bibitem[{Balkwell and Balswick(1981)}]{balkwell81}
Balkwell, C., and Balswick, J. (1981). \bibenquote{Subsistence economy, family
  structure, and the status of the elderly.} \textit{Journal of Marriage and
  the Family}, \textit{43}(2), 423--429.

\bibitem[{Bisin and Verdier(2000)}]{bisver00}
Bisin, A., and Verdier, T. (2000). \bibenquote{Beyond the melting pot: cultural
  transmission, marriage, and the evolution of ethnic and religious traits.}
  \textit{Quarterly Journal of Economics}, \textit{115}(3), 955--988.

\bibitem[{Bisin and Verdier(2001)}]{bisver01}
Bisin, A., and Verdier, T. (2001). \bibenquote{The economics of cultural
  transmission and the dynamics of preferences.} \textit{Journal of Economic
  Theory}, \textit{97}(2), 298--319.

\bibitem[{Blackburn and Cipriani(2006)}]{blackburn05}
Blackburn, K., and Cipriani, G.~P. (2006). \bibenquote{Intergenerational
  transfers and demographic transition.} \textit{Journal of Development
  Economics}, \textit{78}(1), 191--214.

\bibitem[{Blumberg and Winch(1972)}]{blumberg72}
Blumberg, R.~L., and Winch, R.~F. (1972). \bibenquote{Societal complexity and
  familial complexity: Evidence for the curvilinear hypothesis.}
  \textit{American Journal of Sociology}, \textit{77}(5), 898--920.

\bibitem[{Bogoras(1904)}]{bogoras04}
Bogoras, W. (1904). \textit{The Chukchee}. Leiden: E. J. Brill.

\bibitem[{Boldrin and Jones(2002)}]{boldrin02}
Boldrin, M., and Jones, L.~E. (2002). \bibenquote{Mortality, fertility, and
  saving in a {M}althusian economy.} \textit{Review of Economic Dynamics},
  \textit{5}(4), 775--814.

\bibitem[{Bowles(1998)}]{bowles98}
Bowles, S. (1998). \bibenquote{Endogenous preferences: the cultural
  consequences of markets and other economic institutions.} \textit{Journal of
  Economic Literature}, \textit{36}(1), 75--111.

\bibitem[{Cigno(1993)}]{cigno93}
Cigno, A. (1993). \bibenquote{Intergenerational transfers without altruism:
  family, market, and state.} \textit{European Journal of Political Economy},
  \textit{9}(4), 505--518.

\bibitem[{Cigno et~al.(2020)Cigno, Gioffré, and Luporini}]{cigno20}
Cigno, A., Gioffré, A., and Luporini, A. (2020). \bibenquote{Evolution of
  individual preferences and persistence of family rules.} \textit{Review of
  Economics of the Household}, \textit{1}(Online: pub. September 11).

\bibitem[{Cowgill and Holmes(1971)}]{cowgill71}
Cowgill, D.~O., and Holmes, L.~D. (1971). \textit{Aging and modernization}. New
  York: Appleton-Century-Crofts.

\bibitem[{Cozzi(1998)}]{cozzi98}
Cozzi, G. (1998). \bibenquote{Culture as a bubble.} \textit{Journal of
  Political Economy}, \textit{106}(2), 376--394.

\bibitem[{Cremer et~al.(2013)Cremer, Gahvari, and Pestieau}]{cremer13}
Cremer, H., Gahvari, F., and Pestieau, P. (2013). \bibenquote{Endogenous
  altruism, redistribution, and long-term care.} \textit{B. E. Journal of
  Economic Analysis \& Policy}, \textit{14}(2), 499--524.

\bibitem[{Diamond(1965)}]{diamond65}
Diamond, P.~A. (1965). \bibenquote{National debt in a neoclassical growth
  model.} \textit{American Economic Review}, \textit{55}(5), 1125--1150.

\bibitem[{Dowd(1981)}]{dowd81}
Dowd, J.~J. (1981). \bibenquote{Industrialization and the decline of the aged.}
  \textit{Sociological Focus}, \textit{14}(1), 255--269.

\bibitem[{Gordon(1954)}]{gordon54}
Gordon, H.~S. (1954). \bibenquote{The economic theory of a common-property
  resources: the fishery.} \textit{Journal of Political Economy},
  \textit{62}(2), 124--42.

\bibitem[{Gurven and Kaplan(2009)}]{gurven09}
Gurven, M., and Kaplan, H. (2009). \bibenquote{Beyond the grandmother
  hypothesis: evolutionary models of human longevity.} In J.~Sokolovsky (Ed.),
  \textit{The Cultural Context of Aging: Worldwide perspectives}, 53--66,
  Westport, CT: Praeger.

\bibitem[{Holmes(1976)}]{holmes76}
Holmes, L.~D. (1976). \bibenquote{Trends in anthropological gerontology: From
  simmons to the seventies.} \textit{International Journal of Aging and Human
  Development}, \textit{7}(3), 211--220.

\bibitem[{Kaplan and Robson(2002)}]{kaplanrobson02}
Kaplan, H.~S., and Robson, A.~J. (2002). \bibenquote{The emergence of humans:
  The coevolution of intelligence and longevity with intergenerational
  transfers.} \textit{Proceedings of the National Academy of Sciences},
  \textit{99}(15).

\bibitem[{Kaplan and Robson(2003)}]{kaplanrobson03}
Kaplan, H.~S., and Robson, A.~J. (2003). \bibenquote{The evolution of human
  life expectancy and intelligence in hunter-gatherer economies.} \textit{The
  American Economic Review}, \textit{99}(1), 150--169.

\bibitem[{Kelly(2013)}]{kelly13}
Kelly, R.~L. (2013). \textit{The lifeways of hunter-gatherers: the foraging
  spectrum}. New York: Cambridge University Press.

\bibitem[{Klimaviciute et~al.(2017)Klimaviciute, Perlman, Pestieau, and
  Schoenmaeckers}]{klima17}
Klimaviciute, J., Perlman, S., Pestieau, P., and Schoenmaeckers, J. (2017).
  \bibenquote{Caring for dependent parents: altruism, exchange, or family
  norm?} \textit{Journal of Population Economics}, \textit{30}(3), 835--873.

\bibitem[{Ko and M\"{o}ring(2021)}]{ko21}
Ko, P.-C., and M\"{o}ring, K. (2021). \bibenquote{Chipping in or crowding out?
  the impact of pension receipt on older adults' intergenerational support and
  subjective well-being in china.} \textit{Journal of Cross-Cultural
  Gerontology}, \textit{36}(2), 139--154.

\bibitem[{Koda and Uruyos(2018)}]{koda18}
Koda, Y., and Uruyos, M. (2018). \bibenquote{Intergenerational transfers,
  demographic transition, and altruism: Problems in developing asia.}
  \textit{Review of Development Economics}, \textit{22}(3), 904--927.

\bibitem[{Kodate and Timonen(2017)}]{kodate17}
Kodate, N., and Timonen, V. (2017). \bibenquote{Bringing the family in through
  the back door: the stealthy expansion of family care in asian and european
  long-term care policy.} \textit{Journal of Cross-Cultural Gerontology},
  \textit{36}(3), 291--301.

\bibitem[{Laferr\'ere(1999)}]{laferrere99}
Laferr\'ere, A. (1999). \bibenquote{Intergenerational transmissions models: A
  survey.} \textit{The Geneva Papers on Risk and Insurance}, \textit{24}(1),
  2--26.

\bibitem[{Lee(1996)}]{lee96}
Lee, G.~R. (1996). \bibenquote{Economies and families: A further investigation
  of the curvilinear hypothesis.} \textit{Journal of Comparative Family
  studies}, \textit{27}(2), 353--372.

\bibitem[{Lee and Kezis(1981)}]{lee81}
Lee, G.~R., and Kezis, M. (1981). \bibenquote{Societal literacy and the status
  of the aged.} \textit{International Journal of Aging and Human Development},
  \textit{12}(3), 221--234.

\bibitem[{Lee(1979)}]{lee79}
Lee, R.~B. (1979). \textit{The !Kung San: Men, Women and Work in a Foraging
  Society}. cambridge and New York: Cambridge University Press.

\bibitem[{Lee(2003)}]{rlee03}
Lee, R.~D. (2003). \bibenquote{Demographic change, welfare, and
  intergenerational transfers: a global overview.} \textit{Genus},
  \textit{59}(3/4), 43--70.

\bibitem[{McArdle and Yeracaris(1981)}]{mcye81}
McArdle, J.~L., and Yeracaris, C. (1981). \bibenquote{Respect for the elderly
  in preindustrial societies as related to their activity.} \textit{Behavior
  science research}, \textit{16}(3-4), 308--339.

\bibitem[{Michel et~al.(2006)Michel, Thibault, and Vidal}]{michel06}
Michel, P., Thibault, E., and Vidal, J.-P. (2006). \bibenquote{Chapter 15:
  Intergenerational altruism and neoclassical growth models.} In S.-C. Kolm,
  and J.~M. Ythier (Eds.), \textit{Handbook of the Economics of Giving,
  Altruism, and Reciprocity, Volume 2}, 1055--1106, Elsevier.

\bibitem[{Modigliani(1966)}]{modigliani66}
Modigliani, F. (1966). \bibenquote{The life cycle hypothesis of saving, the
  demand for wealth and the supply of capital.} \textit{Social Research},
  \textit{33}(2), 160--217.

\bibitem[{Morand(1999)}]{morand99}
Morand, O. (1999). \bibenquote{Endogenous fertility, income distribution, and
  growth.} \textit{Journal of Economic Growth}, \textit{4}(3), 331--349.

\bibitem[{Murdock(1969)}]{murdock69}
Murdock, G.~P. (1969). \textit{Ethnographic Atlas}. Pittsburg, Penn.:
  University of Pittsburg Press.

\bibitem[{Murdock and Provost(1973)}]{murdock73}
Murdock, G.~P., and Provost, C. (1973). \bibenquote{Measurement of cultural
  complexity.} \textit{Ethnology}, \textit{12}(4), 379--92.

\bibitem[{Murdock and White(1969)}]{murdockwhite69}
Murdock, G.~P., and White, D.~R. (1969). \bibenquote{Standard cross-cultural
  sample.} \textit{Ethnology}, \textit{8}(4), 329--369.

\bibitem[{Nimkoff and Middleton(1960)}]{nimkoff60}
Nimkoff, M.~F., and Middleton, R. (1960). \bibenquote{Types of family and type
  of economy.} \textit{American Journal of Sociology}, \textit{66}(3),
  215--225.

\bibitem[{Nishimura and Zhang(1995)}]{nishzhang95}
Nishimura, K., and Zhang, J. (1995). \bibenquote{Sustainable plans of social
  security with endogenous fertility.} \textit{Oxford Economic Papers},
  \textit{47}(1), 182--94.

\bibitem[{Olivera(2013)}]{olivera13}
Olivera, J. (2013). \textit{Old-age support and demographic transition in
  developing countries: A cultural transmission model}. UCD Geary Institute:
  Discussion Paper Series WP2013/07.

\bibitem[{Ponthiere(2013)}]{ponthiere13}
Ponthiere, G. (2013). \bibenquote{Long-term care, altruism, and socialization.}
  \textit{B. E. Journal of Economic Analysis \& Policy}, \textit{14}(2),
  429--471.

\bibitem[{Press and McKool~Jr.(1972)}]{press72}
Press, I., and McKool~Jr., M. (1972). \bibenquote{Social structure and status
  of the aged: Toward some valid cross-cultural generalizations.} \textit{Aging
  and Human Development}, \textit{3}(4), 297--306.

\bibitem[{Pryor(1973)}]{pryor73}
Pryor, F. (1973). \textit{The Origins of the Economy}. New York, New York:
  Academic Press.

\bibitem[{Rosenberg(2009)}]{rosenberg09}
Rosenberg, H.~G. (2009). \bibenquote{Complaint discourse, aging, and caregiving
  among the {Ju/'hoansi} of botswana.} In J.~Sokolovsky (Ed.), \textit{The
  Cultural Context of Aging: Worldwide perspectives}, 30--52, Westport, CT:
  Praeger.

\bibitem[{Samuelson(1958)}]{samuelson58}
Samuelson, P.~A. (1958). \bibenquote{An exact consumption-loan model of
  interest with or without the social contrivance of money.} \textit{Journal of
  Political Economy}, \textit{66}(6), 467--482.

\bibitem[{Silverman and Maxwell(1978)}]{silverman78}
Silverman, P., and Maxwell, R.~J. (1978). \bibenquote{How do {I} respect thee?
  {L}et me count the ways: Deference towards elderly men and women.}
  \textit{Behavior Science Research}, \textit{13}(2), 91--108.

\bibitem[{Simmons(1962)}]{simmons62}
Simmons, L.~W. (1962). \bibenquote{Aging in primitive societies: A comparative
  survey of family life and relationships.} \textit{Law and Contemporary
  Problems}, \textit{27}(1), 36--51.

\bibitem[{Simmons(1970)}]{simmons70}
Simmons, L.~W. (1970). \textit{The Role of the Aged in Primitive Society}. New
  Haven, Conn.: Yale University Press.

\bibitem[{{United Nations} Department~of Economic {and}
  Social~Affairs(2019)}]{un19}
{United Nations} Department~of Economic {and} Social~Affairs, P.~D. (2019).
  \textit{World Population Ageing 2019: Highlights}. New York: United Nations,
  \url{https://www.un.org/en/development/desa/population/publications/pdf/ageing/WorldPopulationAgeing2019-Highlights.pdf}.

\bibitem[{Varvarigos(2021)}]{varvarigos21}
Varvarigos, D. (2021). \bibenquote{Upstream intergenerational transfers in
  economic development: the role of family ties and their cultural
  transmission.} \textit{Journal of Mathematical Economics}, \textit{96}(1),
  1--13.

\bibitem[{{World Health Organization}(2011)}]{who11}
{World Health Organization} (2011). \textit{Global Health and Aging}. 11-7737,
  National Institute on Aging,
  \url{https://www.who.int/ageing/publications/global_health.pdf}.

\end{thebibliography}
